\documentclass[journal]{IEEEtran}

\usepackage{amsmath,amssymb,amsthm,mathtools}
\usepackage{tikz,graphicx,booktabs,array,arydshln}
\usepackage{url}
\usetikzlibrary{arrows.meta}

\theoremstyle{plain}
\newtheorem{theorem}{Theorem}
\newtheorem{proposition}{Proposition}
\newtheorem{lemma}{Lemma}
\newtheorem{corollary}{Corollary}
\newtheorem{definition}{Definition}
\newtheoremstyle{remstyle}{\topsep}{\topsep}{}{}{\bfseries}{.}{.5em}{}
\theoremstyle{remstyle}

\DeclareMathOperator{\dist}{dist}
\DeclareMathOperator{\spn}{span}
\newcommand{\dd}{\mathrm{d}}
\newcommand{\Lap}{\triangle_\alpha}
\newcommand{\Ha}{\mathcal H_\alpha}
\newcommand{\ip}[2]{\langle #1,#2\rangle}
\newcommand{\cA}{\mathcal A}
\newcommand{\cS}{\mathcal S}
\newcommand{\cX}{\mathcal X}
\newcommand{\BF}{B_{\mathrm F}}
\newcommand{\Pra}{\mathrm{Pr}}
\newcommand{\Ra}{\mathrm{Ra}}

\title{Rayleigh--B\'enard Stabilization\\ by Fredholm Backstepping}
\author{Miroslav Krstic,~\IEEEmembership{Fellow,~IEEE},
and Eugene Lavretsky,~\IEEEmembership{Fellow,~IEEE}%
\thanks{This work was funded by AFOSR grant FA9550-23-1-0535 and NSF grant ECCS-2151525.}%
\thanks{M. Krstic is with the Department of Mechanical and Aerospace Engineering, University of California, San Diego, La Jolla, CA 92093-0411, USA (e-mail: mkrstic@ucsd.edu).}%
\thanks{E. Lavretsky is with The Boeing Company, Huntington Beach, CA 92647, USA (e-mail: eugene.lavretsky@boeing.com).}}

\begin{document}
\maketitle

\begin{abstract}
Among Navier--Stokes systems, the Rayleigh--B\'enard layer, a fluid between two horizontal plates heated from below and cooled from above, is the benchmark of hydrodynamic instability that manifests itself as the unsteady convection phenomenon, which in turn can be interpreted as the onset of chaos: the Lorenz equations represent the Fourier--Galerkin approximation coefficients for the minimal three-mode truncation of it. Its instability has no mean flow behind it. The fluid is at rest, the velocity and the temperature are each stable alone, and buoyancy closes a positive loop between them. The linearization is self-adjoint, its unstable spectrum real, and the number of unstable modes grows without bound in the Rayleigh number; at $\Ra=10^6$ the layer has $85$ of them, with growth rates up to $2137$. We stabilize the conductive state at every Rayleigh and Prandtl number, with any prescribed decay rate, by explicit feedback of the full state to the tangential velocity, temperature, and normal velocity of the upper plate. The feedback laws are obtained by Fredholm backstepping of the coupled velocity--temperature system, with pole placement for the lowest modes.
\end{abstract}


\section{Introduction}

\subsection{Problem}

Fredholm backstepping has now been applied to three plants in the same geometry, a layer of fluid between two plates: the Poiseuille channel \cite{NSF26}, the channel carrying two fluids separated by an interface \cite{Yih26}, and the Rayleigh--B\'enard layer of this paper. The geometry is shared and the method is shared. The physics is not. In the first the instability comes from the shear of the mean flow, in the second from the jump of viscosity across the interface, and here from buoyancy, with the fluid at rest. Each of the three is a benchmark of hydrodynamic stability, and each carries a difficulty the other two do not.

A layer of fluid is held between two horizontal plates and heated from below. Below a critical Rayleigh number it stays at rest, and the temperature falls linearly from the hot plate to the cold one. Above that number, warm fluid rises, cold fluid sinks, and convection rolls set in. This is Rayleigh's problem \cite{Rayleigh1916}, and the rolls are what feedback is asked to prevent.

The linearization has two subsystems. The wall-normal velocity is governed by the Stokes operator, the temperature deviation by the heat equation, and neither is unstable alone. Buoyancy drives the velocity through the horizontal gradient of the temperature, and vertical motion carries heat across the background gradient. The instability is the loop closed by these two couplings. There is no mean flow and no shear, so none of the mechanisms of channel-flow control are present.

The control problem grows with the Rayleigh number. At $\Ra=10^6$, the value used in Section~\ref{sec:sim}, the layer has $85$ unstable eigenvalues over $20$ wavenumbers, with growth rates up to $2137$, and the count grows without bound in $\Ra$.

All actuation is at the upper plate: its temperature, its tangential velocity, its wall-normal velocity. The lower plate is held at the hot temperature, with no slip. One feature of the plant decides the design. Buoyancy enters the velocity equation through the inverse Laplacian, so the coupling is nonlocal across the layer. A Volterra transformation integrates from one wall to the current position, and against a nonlocal coupling it has nothing to work with; a Fredholm transformation integrates over the whole layer, which is what this plant requires. What is asked of the feedback is not stabilization alone, but a decay rate specified in advance and achieved at every Rayleigh and Prandtl number.

\subsection{Literature}

The linear stability of the heated layer is Rayleigh's \cite{Rayleigh1916}, with rigid walls in Pellew and Southwell \cite{PS40}, who computed the critical Rayleigh number $1708$. Chandrasekhar \cite{Chandra61} and Drazin and Reid \cite{DR04} are the standard accounts.

The layer has been controlled by feedback since Tang and Bau \cite{TB93}, who shifted the onset of convection to a higher Rayleigh number by heating the lower plate in proportion to measured temperatures. Howle \cite{Howle97} did the same in experiment, with shadowgraph measurements. Or, Cortelezzi and Speyer \cite{OCS01} designed linear-quadratic-Gaussian and $\mathcal H_\infty$ controllers for the linearized Boussinesq system and held the state of rest up to $14.5$ times the critical Rayleigh number. These designs are finite-dimensional, built on a truncation or on a few measured modes.

In the PDE setting, Vazquez and Krstic \cite{VK06} stabilized a thermal convection loop by backstepping: a singular perturbation in the Prandtl number leaves a single equation for the temperature, a Volterra design applies to it, and exponential stability follows for large Prandtl number. Ramaswamy, Raymond and Roy \cite{RRR19} stabilize the Boussinesq system by finite-dimensional boundary feedback.

For the channel, boundary stabilization at arbitrary Reynolds number is by Volterra backstepping in \cite{VK07,VK08} and by spectral finite-dimensional feedback in \cite{Barbu2007,Munteanu2011,Munteanu2012}. The Fredholm design for the channel is \cite{NSF26} and for the two-fluid channel \cite{Yih26}.

A Fredholm transformation was first used for boundary control in \cite{SGK2009}, on the Euler--Bernoulli beam, where no Volterra transformation assigns an arbitrary decay rate. What turned that into a method is the work of Coron and Lü \cite{CoronLu2014,CoronLu2015}, on the Korteweg--de Vries and Kuramoto--Sivashinsky equations, where the transformation is built from the controllability of the plant and the spectral properties of the operator, and rapid stabilization is obtained from those properties rather than from a kernel equation. That line continues in the treatment of critical operators by Gagnon, Hayat, Xiang and Zhang \cite{GHXZ25}, in the abstract single-input theorem of Hayat and Loko \cite{HayatLoko26}, and in the parabolic case of Boulard and Hayat \cite{BoulardHayat2025}. The present design follows it.

\subsection{Contributions}

\begin{itemize}
\item\emph{Rapid stabilization of the Rayleigh--B\'enard layer at every Rayleigh and Prandtl number.} We stabilize the conductive state and assign its decay rate, by feedback to the temperature and the two velocity components of the upper plate. We remove the two-way coupling that destabilizes two fields each stable alone, by feedback rather than by the large-Prandtl singular perturbation the earlier PDE design depended on \cite{VK06}, so no Prandtl number is excluded. Where earlier feedback moved the onset of convection \cite{TB93}, \cite{Howle97} or held the rest state to $\Ra\approx25{,}000$ on a truncation \cite{OCS01}, we shift the whole spectrum of every actuated wavenumber at once --- infinitely many eigenvalues, at $\Ra=1{,}000{,}000$ the $85$ unstable among them, with growth rates to $2137$.
\end{itemize}

To advance the Fredholm design to the Rayleigh--B\'enard stabilization, we have resolved two technical problems that do not arise in the Poiseuille channel \cite{NSF26} or the two-fluid channel \cite{Yih26}.

\begin{itemize}
\item\emph{Eigenvalues that no design parameter separates are assigned two at a time.} A Fredholm transformation assigns modes one at a time and needs them apart --- held apart by a pre-feedback in the Poiseuille channel \cite{NSF26}, by the coupling between the layers in the two-fluid channel \cite{Yih26}. Here we find the velocity and the thermal eigenvalues approaching each other in pairs, and at unit Prandtl number a velocity eigenvalue and a thermal eigenvalue lying within a fixed distance of each other however high in the spectrum one goes. We assign them two at a time, one wall input per field, and prove the $2\times2$ matrix of each pair invertible however close its two eigenvalues come.
\item\emph{The Fredholm transformation is made invertible by an arbitrarily small deformation of the assigned poles.} In the Poiseuille channel \cite{NSF26} and the two-fluid channel \cite{Yih26} the transformation is invertible because the states it prescribes are close to a basis, and nothing further is needed. For the Rayleigh--B\'enard flow pairs leave a finite matrix that can be singular at the very poles one asks for, so closeness no longer suffices. We deform the assigned poles continuously away from the plant's own eigenvalues and prove the Fredholm transformation invertible except at finitely many values of the deformation parameter, which is why any rate is reached once finitely many poles are moved by an arbitrarily small amount.
\end{itemize}

\section{Rayleigh--B\'enard layer and its wall inputs}\label{sec:model}

A plane layer of fluid, $(x,y)\in\mathbb R\times[0,1]$, with gravity along $-y$, is heated from below. In the Oberbeck--Boussinesq approximation of the Navier--Stokes equations, in units of the layer height, the thermal diffusion time and the temperature difference between the walls, the velocity $(u,v)$, the pressure $p$ and the temperature $T$ satisfy
\begin{subequations}\label{eq:NS}
\begin{align}
u_t+uu_x+vu_y&=-p_x+\Pra\Delta u,\label{eq:NSu}\\
v_t+uv_x+vv_y&=-p_y+\Pra\Delta v+\Pra\Ra\,T,\label{eq:NSv}\\
0&=u_x+v_y,\label{eq:NSdiv}\\
T_t+uT_x+vT_y&=\Delta T,\label{eq:NST}
\end{align}
\end{subequations}
with all fields periodic in $x$ with period $L$,
\begin{equation}\label{eq:per}
(u,v,p,T)(t,x+L,y)=(u,v,p,T)(t,x,y),
\end{equation}
no slip and the hot wall at $y=0$, and the three inputs at $y=1$,
\begin{subequations}\label{eq:NSbc}
\begin{align}
(u,\ v,\ T)&=(0,\ 0,\ 1)\qquad\text{at }y=0,\label{eq:NSbc0}\\
(u,\ T,\ v)&=(U_c,\ T_c,\ V_c)\qquad\text{at }y=1,\label{eq:NSbc1}
\end{align}
\end{subequations}
the tangential velocity, the temperature, and the wall-normal velocity of the upper wall. With zero inputs, $u=v=0$, $\bar T=1-y$ and $\bar p=\Pra\Ra\,(y-y^2/2)$ is the equilibrium of the RB layer. With $\theta\coloneqq T-\bar T$, the stream function $\psi$, $u=\psi_y$, $v=-\psi_x$, which satisfies \eqref{eq:NSdiv}, and the vorticity $\omega\coloneqq u_y-v_x=\Delta\psi$, the curl of \eqref{eq:NSu}--\eqref{eq:NSv} eliminates the pressure, and \eqref{eq:NS}--\eqref{eq:NSbc} take the form
\begin{subequations}\label{eq:nl}
\begin{align}
\omega_t+\psi_y\omega_x-\psi_x\omega_y&=\Pra\Delta\omega-\Pra\Ra\,\theta_x,\label{eq:vort}\\
\theta_t+\psi_y\theta_x-\psi_x\theta_y&=\Delta\theta-\psi_x,\label{eq:temp}\\
\psi=\psi_y=\theta&=0\qquad\text{at }y=0,\label{eq:bclow}\\
\big(\psi_y,\ \theta,\ -\psi_x\big)&=\big(U_c,\ T_c,\ V_c\big)\qquad\text{at }y=1,\label{eq:bcphys}\\
(\psi,\theta)(t,x+L,y)&=(\psi,\theta)(t,x,y),\label{eq:perpsi}
\end{align}
\end{subequations}
where $\psi=0$ at $y=0$ in \eqref{eq:bclow} fixes the additive constant of $\psi$, and \eqref{eq:perpsi} follows from \eqref{eq:per} since $\psi(t,x,y)=\int_0^yu(t,x,\eta)\,\dd\eta$. Figure~\ref{fig:layer} shows the configuration.

\begin{figure*}[t]
\centering
\resizebox{0.77\textwidth}{!}{%
\begin{tikzpicture}[x=1cm,y=1cm,>={Stealth[length=2mm]}]
\draw[gray!55,dashed,line width=0.3pt] (0,4.45) -- (11,4.45);
\draw[line width=0.7pt] plot coordinates{(0.000,4.567) (0.069,4.570) (0.138,4.575) (0.206,4.579) (0.275,4.582) (0.344,4.582) (0.413,4.578) (0.481,4.570) (0.550,4.556) (0.619,4.538) (0.688,4.515) (0.756,4.490) (0.825,4.464) (0.894,4.438) (0.963,4.415) (1.031,4.397) (1.100,4.383) (1.169,4.377) (1.238,4.376) (1.306,4.381) (1.375,4.390) (1.444,4.402) (1.513,4.416) (1.581,4.428) (1.650,4.439) (1.719,4.446) (1.788,4.450) (1.856,4.451) (1.925,4.448) (1.994,4.444) (2.062,4.440) (2.131,4.437) (2.200,4.437) (2.269,4.440) (2.338,4.449) (2.406,4.462) (2.475,4.478) (2.544,4.497) (2.613,4.516) (2.681,4.533) (2.750,4.546) (2.819,4.553) (2.888,4.552) (2.956,4.543) (3.025,4.526) (3.094,4.501) (3.163,4.471) (3.231,4.437) (3.300,4.402) (3.369,4.370) (3.438,4.343) (3.506,4.324) (3.575,4.315) (3.644,4.316) (3.713,4.327) (3.781,4.348) (3.850,4.377) (3.919,4.410) (3.988,4.446) (4.056,4.482) (4.125,4.514) (4.194,4.540) (4.263,4.560) (4.331,4.571) (4.400,4.575) (4.469,4.573) (4.538,4.564) (4.606,4.552) (4.675,4.539) (4.744,4.525) (4.812,4.514) (4.881,4.505) (4.950,4.499) (5.019,4.497) (5.088,4.497) (5.156,4.500) (5.225,4.502) (5.294,4.504) (5.363,4.504) (5.431,4.501) (5.500,4.495) (5.569,4.487) (5.638,4.477) (5.706,4.465) (5.775,4.455) (5.844,4.446) (5.912,4.441) (5.981,4.439) (6.050,4.442) (6.119,4.450) (6.188,4.460) (6.256,4.473) (6.325,4.487) (6.394,4.499) (6.463,4.509) (6.531,4.514) (6.600,4.514) (6.669,4.507) (6.738,4.495) (6.806,4.477) (6.875,4.455) (6.944,4.430) (7.013,4.404) (7.081,4.380) (7.150,4.358) (7.219,4.340) (7.288,4.327) (7.356,4.317) (7.425,4.312) (7.494,4.309) (7.562,4.307) (7.631,4.304) (7.700,4.300) (7.769,4.292) (7.838,4.282) (7.906,4.268) (7.975,4.251) (8.044,4.233) (8.113,4.216) (8.181,4.201) (8.250,4.192) (8.319,4.190) (8.388,4.196) (8.456,4.212) (8.525,4.237) (8.594,4.270) (8.663,4.310) (8.731,4.354) (8.800,4.399) (8.869,4.443) (8.938,4.481) (9.006,4.513) (9.075,4.535) (9.144,4.547) (9.213,4.548) (9.281,4.541) (9.350,4.526) (9.419,4.506) (9.488,4.483) (9.556,4.461) (9.625,4.442) (9.694,4.428) (9.763,4.420) (9.831,4.420) (9.900,4.427) (9.969,4.441) (10.038,4.459) (10.106,4.479) (10.175,4.501) (10.244,4.522) (10.312,4.540) (10.381,4.554) (10.450,4.564) (10.519,4.570) (10.588,4.572) (10.656,4.571) (10.725,4.569) (10.794,4.567) (10.863,4.565) (10.931,4.565) (11.000,4.567)};
\draw[gray!70,line width=0.3pt] (1.306,4.381) -- (1.306,4);
\draw[gray!70,line width=0.3pt] (2.819,4.553) -- (2.819,4);
\draw[gray!70,line width=0.3pt] (4.400,4.575) -- (4.400,4);
\draw[gray!70,line width=0.3pt] (6.050,4.442) -- (6.050,4);
\draw[gray!70,line width=0.3pt] (7.425,4.312) -- (7.425,4);
\draw[gray!70,line width=0.3pt] (9.006,4.513) -- (9.006,4);
\draw[gray!70,line width=0.3pt] (10.175,4.501) -- (10.175,4);
\draw[gray!55,line width=0.3pt] (0,0) rectangle (11,4);
\draw[line width=0.9pt] (0,4) -- (11,4);
\draw[line width=0.9pt] (0,0) -- (11,0);
\draw[->,line width=0.5pt] (0.7,3.6) -- (2.2,3.6);
\draw[->,line width=0.5pt] (4.1,3.6) -- (3.2,3.6);
\draw[->,line width=0.5pt] (5.4,3.6) -- (7.5,3.6);
\draw[->,line width=0.5pt] (10.3,3.6) -- (8.8,3.6);
\node[anchor=west,font=\footnotesize] at (0,3.1) {$\psi_y(1)=U_c(x,t),\quad -\psi_x(1)=V_c(x,t)$};
\node[anchor=east,font=\footnotesize] at (11,5.15) {$\theta(1)=T_c(x,t)$, imposed wall temperature};
\node[font=\small] at (5.5,2.25)
  {$\omega_t+\psi_y\omega_x-\psi_x\omega_y=\Pra\Delta\omega-\Pra\Ra\,\theta_x$};
\node[font=\small] at (5.5,1.6)
  {$\theta_t+\psi_y\theta_x-\psi_x\theta_y=\Delta\theta-\psi_x$};
\node[font=\footnotesize] at (5.5,0.95) {$\omega=\Delta\psi$,\;\; $u=\psi_y$,\;\; $v=-\psi_x$,\;\; $\bar T=1-y$};
\node[anchor=west,font=\footnotesize] at (0,-0.45) {$y=0$, heated from below};
\node[anchor=east,font=\footnotesize] at (11,-0.45) {$\psi=\psi_y=\theta=0$};
\node[font=\footnotesize] at (5.5,-1.05) {periodic in $x$, period $L$};
\node[font=\footnotesize] at (11.95,2.75) {gravity};
\draw[->,line width=0.5pt] (11.95,2.45) -- (11.95,1.25);
\end{tikzpicture}}
\caption{The RB layer with its wall inputs. The tangential velocity is a vector and is drawn as arrows; the wall temperature is a scalar datum and is drawn as its profile above the wall.}
\label{fig:layer}
\end{figure*}
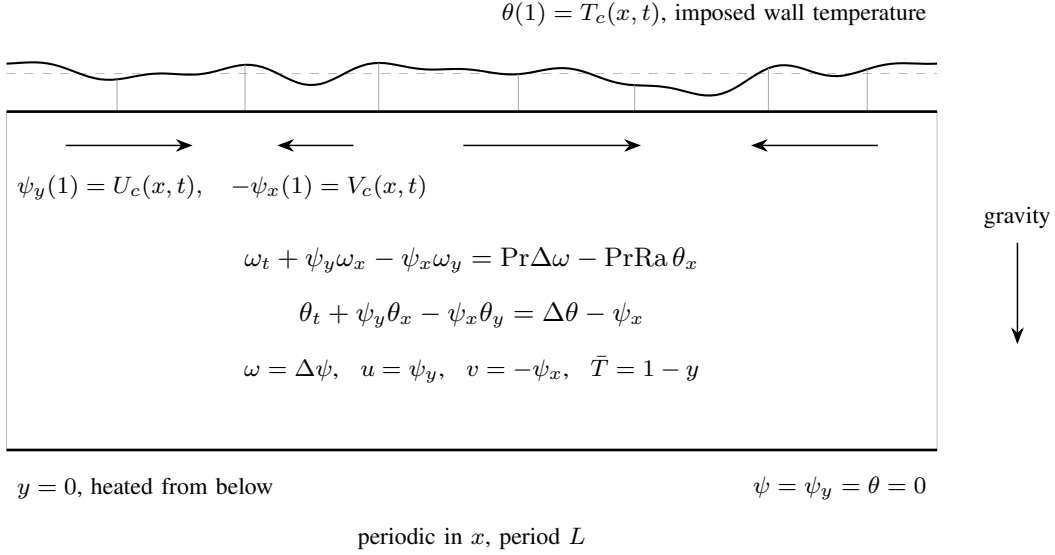

Dropping the quadratic terms and taking the Fourier transform in $x$, with $\alpha=2\pi m/L$, $m\in\mathbb Z$, the Laplacian at the wavenumber $\alpha$,
\begin{equation}\label{eq:Lap}
\Lap\coloneqq\partial_{yy}-\alpha^2,
\end{equation}
and $V\coloneqq-i\alpha\hat\psi=\hat v$, gives per wavenumber $\alpha\ne0$
\begin{subequations}\label{eq:plant}
\begin{align}
\Lap V_t&=\Pra\Lap^2V-\Pra\Ra\,\alpha^2\theta,\label{eq:pV}\\
\theta_t&=\Lap\theta+V,\label{eq:ptheta}\\
0&=V(0)=V_y(0)=\theta(0),\label{eq:pbc}\\
\big(V(1),\ V_y(1),\ \theta(1)\big)&=\big(V_c,\ -i\alpha U_c,\ T_c\big),\label{eq:pbc2}
\end{align}
\end{subequations}
with the state at wavenumber $\alpha$ and the tangential velocity
\begin{equation}\label{eq:state}
X\coloneqq(V,\theta),\qquad u=\frac{V_y}{-i\alpha}.
\end{equation}
The states of \eqref{eq:plant} are the wall-normal velocity and the temperature of \eqref{eq:NS}; the stream function and the vorticity of \eqref{eq:nl} serve only to eliminate the pressure. The velocity subsystem is fourth order in $y$, clamped at $y=0$ with two free data at $y=1$; the thermal subsystem is second order, Dirichlet at $y=0$ with one free datum at $y=1$. Since $\Delta v=-\partial_x\omega$ for every two-dimensional incompressible field, the operator $\Lap$ acting on $V_t$ in \eqref{eq:pV} is the Biot--Savart map at wavenumber $\alpha$, and the coupling of $\theta$ into $V$ is through $\Lap^{-1}$, an integral over the whole layer.

Throughout, $\alpha>0$ (the wavenumber $-\alpha$ follows by conjugation), $\ip\cdot\cdot$ is the $L^2(0,1)$ inner product, $\ip\cdot\cdot_\cX$ and $\ip\cdot\cdot_\alpha$ those of $\cX$ and $\Ha$, $\Ha=H^1_0(0,1)$ with $\|V\|_\alpha^2=\|V_y\|^2+\alpha^2\|V\|^2$, so that $\|u\|^2+\|V\|^2=\|V\|_\alpha^2/\alpha^2$.

\begin{definition}[State space and operators]\label{def:ops}
$\cX\coloneqq\Ha\times L^2(0,1)$ with $\|(V,\theta)\|_\cX^2=\|V\|_\alpha^2+\Pra\Ra\,\alpha^2\|\theta\|^2$. With $\mathcal D\coloneqq H^3\cap H^2_0(0,1)$ and $\mathcal D_\theta\coloneqq H^2\cap H^1_0(0,1)$,
\begin{subequations}\label{eq:ops}
\begin{gather}
S_\alpha\coloneqq\Pra\,\Lap^{-1}\Lap^2\ \text{on }\mathcal D,\label{eq:S}\\
\cA\coloneqq\begin{pmatrix}S_\alpha&-\Pra\Ra\,\alpha^2\Lap^{-1}\\[2pt]1&\Lap\end{pmatrix},\label{eq:cA}\\
\cS\coloneqq\begin{pmatrix}S_\alpha&0\\0&\Lap\end{pmatrix},\label{eq:cS}
\end{gather}
\end{subequations}
with $\mathcal D(\cA)=\mathcal D(\cS)=\mathcal D\times\mathcal D_\theta$.
\end{definition}

In \eqref{eq:S} the inverse $\Lap^{-1}$ is taken with the Dirichlet conditions of $\Ha$, while $\Lap^2$ acts on $\mathcal D$, whose elements satisfy $V=V_y=0$ at both walls, so the two do not cancel: $S_\alpha$ is the clamped Stokes operator, of fourth order, and not $\Pra\Lap$.

\section{Controller and main result}\label{sec:main}

\subsection{Controller and stabilization theorem}\label{sec:gains}

The controller is a full-state feedback: at each wavenumber $\alpha$ in the controlled band it integrates the profiles $u(\cdot,t)$, $V(\cdot,t)$, $\theta(\cdot,t)$ against gain functions of $y$, computed offline from the plant data by Tables~\ref{tab:static} and~\ref{tab:fred}. With $g^u,g^V,g^\theta\in L^2(0,1;\mathbb C^3)$, rows in the order of the inputs $V_c,U_c,T_c$,
\begin{equation}\label{eq:controller}
\begin{split}
\begin{pmatrix}V_c(\alpha,t)\\[2pt]U_c(\alpha,t)\\[2pt]T_c(\alpha,t)\end{pmatrix}
&=\int_0^1\Big[g^u(y)\,u(\alpha,y,t)+g^V(y)\,V(\alpha,y,t)\\
&\qquad\quad+g^\theta(y)\,\theta(\alpha,y,t)\Big]\dd y .
\end{split}
\end{equation}
Written with functionals, \eqref{eq:controller} is
\begin{equation}\label{eq:fb}
V_c=K_V[X],\qquad U_c=K_U[X],\qquad T_c=K_T[X],
\end{equation}
where each gain functional is a sum of two parts, $K_V=\kappa_V+F_V$, $K_U=\kappa_U+F_U$, $K_T=\kappa_T+F_T$, produced by the two stages of the design: the parts $\kappa_V,\kappa_U,\kappa_T$ place the eigenvalues of the finitely many lowest modes of the plant, by Table~\ref{tab:static}, and the parts $F_V,F_U,F_T$ are the Fredholm gains that assign the remaining eigenvalues, by Table~\ref{tab:fred}. Only $F_U$ and $F_T$ act on infinitely many modes; $\kappa_V,\kappa_U,\kappa_T$ and $F_V$ act on finitely many.

At $\alpha=0$ the mean tangential velocity and the mean temperature are fed back through the two wall inputs,
\begin{subequations}\label{eq:meanfbsys}
\begin{gather}
T_{c,0}=\int_0^1k(1,\eta)\,\bar\theta(\eta)\,\dd\eta,\label{eq:meanfb}\\
k(y,\eta)=-q_0\,\eta\,\frac{I_1\!\big(\sqrt{q_0(y^2-\eta^2)}\big)}{\sqrt{q_0(y^2-\eta^2)}},\label{eq:meanfbk}\\
U_{c,0}=\int_0^1k_u(1,\eta)\,\bar u(\eta)\,\dd\eta,\label{eq:meanfbu}\\
k_u(y,\eta)=-\frac{q_u}{\Pra}\,\eta\,\frac{I_1\!\big(\sqrt{(q_u/\Pra)(y^2-\eta^2)}\big)}{\sqrt{(q_u/\Pra)(y^2-\eta^2)}},\label{eq:meanfbku}
\end{gather}
\end{subequations}
with the margins $q_0,q_u\ge0$ and the modified Bessel function $I_1$.

The gains are designed so that a $2\times2$ transformation of the Fredholm form
\begin{equation}\label{eq:transf}
\begin{pmatrix}W\\\Theta\end{pmatrix}(y)=\begin{pmatrix}V\\\theta\end{pmatrix}(y)-\int_0^1k(y,\eta)\begin{pmatrix}V_y\\V\\\theta\end{pmatrix}(\eta)\,\dd\eta ,
\end{equation}
with a kernel $k\in L^2\big((0,1)^2;\mathbb C^{2\times3}\big)$ (for $X\in\cX$, $V_y\in L^2(0,1)$, so the integral is a pairing in $L^2(0,1)$; its boundedness on $\cX$ is proved with the design), maps the closed loop onto the shifted Stokes and shifted heat equations, up to finitely many eigenvalues displaced by an arbitrarily small amount,
\begin{subequations}\label{eq:target}
\begin{align}
\Lap W_t&=\Pra\Lap^2W-q_1\Lap W,\label{eq:tW}\\
\Theta_t&=\Lap\Theta-q_2\Theta,\label{eq:tTheta}\\
0&=W(0)=W_y(0)=W(1)=W_y(1)\nonumber\\
&=\Theta(0)=\Theta(1).\label{eq:tbc}
\end{align}
\end{subequations}

\begin{theorem}[Three wall inputs stabilize the RB layer at any decay rate]\label{thm:phys}
Fix $\Pra>0$, $\Ra>0$, $L>0$ and $\omega>0$. Under the feedback \eqref{eq:controller}, \eqref{eq:meanfbsys} with the gains of Tables~\ref{tab:static} and~\ref{tab:fred}, the physical inputs $U_c(x,t)$, $T_c(x,t)$, $V_c(x,t)$ are real for real fields, the linearized closed loop \eqref{eq:NS}--\eqref{eq:NSbc} generates an analytic semigroup on $\mathcal H$, and every solution satisfies
\begin{equation}\label{eq:physdecay}
\begin{split}
&\|u(t)\|^2+\|v(t)\|^2+\|\theta(t)\|^2\\
&\qquad\le C\,e^{-2\omega t}\big(\|u(0)\|^2+\|v(0)\|^2+\|\theta(0)\|^2\big),
\end{split}
\end{equation}
where
\begin{itemize}
\item[$\bullet$] $C$ is the largest of the constants of the actuated wavenumbers, of the two mean transformations, and of the conversion factor $\chi\coloneqq\max(1,\Pra\Ra)/\min(1,\Pra\Ra)$ on the unactuated ones;
\item[$\bullet$] $\mathcal H_m\coloneqq\cX(\alpha_m)$ for $|m|>M$, $M$ of step~0, and, for $0<|m|\le M$, $\mathcal H_m$ is the set of $(V,\theta)\in H^1(0,1)\times L^2(0,1)$ with $V(0)=0$ whose wall value $V(1)$ is the one assigned by the row $K_V$ of \eqref{eq:fb}, both with the norm of Definition~\ref{def:ops};
\item[$\bullet$] $\mathcal H$ consists of the fields whose mean $(\bar u,\bar\theta)$ lies in $L^2(0,1)^2$, whose Fourier coefficient $X_m$ at each $m\ne0$ lies in $\mathcal H_m$ and is the conjugate of $X_{-m}$, and for which $\|(\bar u,\bar\theta)\|^2+\sum_{m\ne0}\alpha_m^{-2}\|X_m\|^2_{\mathcal H_m}$ is finite.
\end{itemize}
\end{theorem}

Theorem~\ref{thm:phys} is proved in Section~\ref{sec:phys}, by assembling the wavenumbers; the design at a single wavenumber, the transformation it produces and the closed-loop spectrum there are Theorem~\ref{thm:main}, at the end of Section~\ref{sec:fredholm}. Step~0 of Table~\ref{tab:static} fixes the actuated band and the mean-field gains; steps 1--9 compute the gains at a fixed wavenumber $\alpha>0$ with $N$ plant modes retained, $N=\infty$ in Theorem~\ref{thm:main}: Table~\ref{tab:static} the plant data and pole placement on the low block, Table~\ref{tab:fred} the Fredholm stage on the clusters and the gain functions.

\begin{table}[t]
\caption{Gains, part 1: band, mean field, plant data and pole placement.}\label{tab:static}
\centering
{\renewcommand{\arraystretch}{1.25}\setlength{\arrayrulewidth}{0.2pt}\begin{tabular}{@{}r@{\hspace{0.7em}}p{\dimexpr\columnwidth-1.9em\relax}@{}}
\toprule
0 & \emph{Band and mean field.} $M$ the least integer such that the plant \eqref{eq:plant} at every $|m|>M$ has no eigenvalue above $-\omega$, finite by Proposition~\ref{prop:mean}(iii). Steps 1--9 at each $0<|m|\le M$, the conjugate gains at $-m$, and zero input at $|m|>M$. At $\alpha=0$, \eqref{eq:meanfbsys} with $q_0\ge\max(0,\omega-\pi^2)$ and $q_u\ge\max(0,\omega-\Pra\pi^2)$.\\[2pt]
\hdashline\noalign{\vskip3pt}
1 & \emph{Plant modes.} Solve $\cA X_j=\lambda_jX_j$, $j=1,\dots,N$, $X_j=(V_j,\theta_j)$, $\lambda_1\ge\lambda_2\ge\cdots$, normalized by $\|V_j\|_\alpha^2+\Pra\Ra\,\alpha^2\|\theta_j\|^2=1$ and orthogonalized in $\cX$ within each multiple eigenvalue. Modal coordinates of a state $X=(V,\theta)$, $-\Lap V=i\alpha\hat\omega$:
\begin{equation}\label{eq:rc-coord}
c_j[X]=i\alpha\int_0^1\hat\omega\,\overline{V_j}\,\dd y+\Pra\Ra\,\alpha^2\int_0^1\theta\,\overline{\theta_j}\,\dd y .
\end{equation}\vspace{-\baselineskip}\\[2pt]
\hdashline\noalign{\vskip3pt}
2 & \emph{Wall rows.} $b_j\coloneqq\big(b_0[V_j],\ b_1[V_j],\ b_2[\theta_j]\big)\in\mathbb C^3$, in the input order $(V_c,U_c,T_c)$, with
\begin{equation}\label{eq:functionals}
\begin{gathered}
b_0[\zeta]\coloneqq\Pra\,\zeta'''(1),\qquad
b_1[\zeta]\coloneqq i\alpha\Pra\,\zeta''(1),\\
b_2[\psi]\coloneqq-\Pra\Ra\,\alpha^2\,\psi'(1).
\end{gathered}
\end{equation}\vspace{-\baselineskip}\\[2pt]
\hdashline\noalign{\vskip3pt}
3 & \emph{Uncoupled modes.} $(\lambda^S_n,e_n)$ the eigenpairs of $S_\alpha$, $\|e_n\|_\alpha=1$; $\delta_m=-(m^2\pi^2+\alpha^2)$ and $\theta_m=\sqrt2\sin(m\pi y)/(\alpha\sqrt{\Pra\Ra})$. The joint list $\varsigma_1\ge\varsigma_2\ge\cdots$ of all $\lambda^S_n$ and $\delta_m$, with uncoupled eigenvectors $\tilde e_k\in\{(e_n,0),(0,\theta_m)\}$ and family $i(k)=1$ for $\lambda^S_n$, $i(k)=2$ for $\delta_m$; $E_{jk}\coloneqq\ip{\tilde e_k}{X_j}_\cX$. Target poles $\nu_k\coloneqq\varsigma_k-q_{i(k)}$, with shifts $q_1,q_2\ge\omega$ chosen so that
\begin{equation}\label{eq:shifts}
q_1-q_2\ne\lambda^S_n-\delta_m,\quad \lambda^S_n-q_1\notin\sigma(\cA),\quad \delta_m-q_2\notin\sigma(\cA)
\end{equation}
for all $n,m\ge1$; such shifts are dense in $[\omega,\infty)^2$ by Proposition~\ref{prop:design}.\\[2pt]
\hdashline\noalign{\vskip3pt}
4 & \emph{Low block.} $\mu_j\coloneqq\nu_j$, $j\le d$. $K_L\in\mathbb C^{3\times d}$ placing the eigenvalues of $\operatorname{diag}(\lambda_1,\dots,\lambda_d)+B_LK_L$ at $\mu_1,\dots,\mu_d$, $B_L$ with rows $b_1,\dots,b_d$; eigenvectors $x_j\in\mathbb C^d$, and $c_j\coloneqq K_Lx_j$.\\[2pt]
\hdashline\noalign{\vskip3pt}
5 & \emph{Pre-compensated basis.} The columns of $W\in\mathbb C^{N\times N}$ are the plant coordinates of the eigenvectors $X'_j$ of the pre-compensated plant, with eigenvalues $\lambda'_j$:
\begin{equation}\label{eq:rc-W}
W_{ij}=\begin{cases}(x_j)_i,&i,j\le d,\\ \dfrac{b_i\cdot c_j}{\mu_j-\lambda_i},&i>d,\ j\le d,\\ \delta_{ij},&j>d,\end{cases}
\qquad\lambda'_j\coloneqq\begin{cases}\mu_j,&j\le d,\\ \lambda_j,&j>d.\end{cases}
\end{equation}
Input rows in the dual basis: $\beta^{(3)}\coloneqq W^{-1}B\in\mathbb C^{N\times3}$, $B$ with rows $b_j$; $\beta$ its last two columns, of $U_c$ and $T_c$.\\
\bottomrule
\end{tabular}}
\end{table}

\begin{table}[t]
\caption{Gains at wavenumber $\alpha$, part 2: Fredholm stage and gain functions.}\label{tab:fred}
\centering
{\renewcommand{\arraystretch}{1.25}\setlength{\arrayrulewidth}{0.2pt}\begin{tabular}{@{}r@{\hspace{0.7em}}p{\dimexpr\columnwidth-1.9em\relax}@{}}
\toprule
6 & \emph{Clusters.} Join consecutive entries $\varsigma_k$, $\varsigma_{k+1}$, $k>d$, of the joint list when $|\varsigma_k-\varsigma_{k+1}|\le\epsilon\min\big(\sqrt{|\varsigma_k|},\sqrt{|\varsigma_{k+1}|}\big)$; the clusters $G$ are the connected components, as sets of positions in the joint list, which are also the plant indices of the cluster. $M_G\coloneqq\operatorname{diag}(\lambda_j)_{j\in G}$, $\mathbb B_G\coloneqq(\beta_j)_{j\in G}$. Clusters of positions below $n_*$, with $n_*$ large enough, are middle, the others tail; for a tail cluster, $\mathbb B_G^{(+)}$ is the inverse of $\mathbb B_G$ for a pair and its minimum-norm right inverse for a singleton.\\[2pt]
\hdashline\noalign{\vskip3pt}
7 & \emph{Closed-loop eigenvectors.} Columns $\Psi_{\cdot k}(\xi)$, in the coordinates of $\{X'_j\}$, and input vectors $c_k(\xi)\in\mathbb C^3$. Tail cluster $G\ni k$, with $v\coloneqq(E_{jk})_{j\in G}$:
\begin{equation}\label{eq:rc-c}
c_k\coloneqq\big(0,\ \mathbb B_G^{(+)}(\nu_k-M_G)v\big),\qquad
\Psi_{jk}\coloneqq\begin{cases}E_{jk},&j\in G,\\[2pt] \dfrac{\beta^{(3)}_j\cdot c_k}{\nu_k-\lambda'_j},&j\notin G,\end{cases}
\end{equation}
and $\nu_k(\xi)\coloneqq\nu_k$. Middle cluster $G$: a bijection $k\mapsto p(k)$ of $G$ onto itself, and $\hat c_k\in\mathbb C^3$ such that the matrix $\big(\beta^{(3)}_j\cdot\hat c_k\big)$, $k\in G$, $\lambda_j=\lambda_{p(k)}$, is invertible for each eigenvalue of $\cA$ in $G$;
\begin{gather}
\nu_k(\xi)\coloneqq\lambda_{p(k)}+\xi\big(\nu_k-\lambda_{p(k)}\big),\label{eq:rc-mid}\\
c_k(\xi)\coloneqq\big(\nu_k(\xi)-\lambda_{p(k)}\big)\hat c_k,\label{eq:rc-midc}\\
\Psi_{jk}(\xi)\coloneqq\frac{\nu_k(\xi)-\lambda_{p(k)}}{\nu_k(\xi)-\lambda'_j}\,\beta^{(3)}_j\cdot\hat c_k,\label{eq:rc-mid2}
\end{gather}
the fraction read as $1$ when $\lambda'_j=\lambda_{p(k)}$.
Low block, $k\le d$: $\Psi_{\cdot k}$ the $k$-th unit vector, $c_k\coloneqq0$, $\nu_k(\xi)\coloneqq\mu_k$.\\[2pt]
\hdashline\noalign{\vskip3pt}
8 & \emph{Fredholm gain.} With $N$ at a cluster boundary and $\xi\in[1,1+\delta)$ such that $\Psi(\xi)$ is invertible, $\delta$ small enough that the displaced poles stay to the left of $-\omega$, $\Xi(\xi)\coloneqq[c_1\cdots c_N]$ and, in plant coordinates,
\begin{equation}\label{eq:rc-F}
F\coloneqq \Xi(\xi)\,\Psi(\xi)^{-1}W^{-1}\in\mathbb C^{3\times N}.
\end{equation}\vspace{-\baselineskip}\\[2pt]
\hdashline\noalign{\vskip3pt}
9 & \emph{Gain functions.} $K\coloneqq[K_L\ \,0]+F$, $(V_c,U_c,T_c)^\top=\sum_jK_{\cdot j}\,c_j[X]$, and, for the rows $r=0,1,2$ of \eqref{eq:controller}, with $\phi_r=(\phi^V_r,\phi^\theta_r)\coloneqq\sum_j\overline{K_{r+1,j}}\,X_j$,
\begin{equation}\label{eq:gains}
g^u_r=-i\alpha\,\overline{(\phi^V_r)'},\qquad g^V_r=\alpha^2\,\overline{\phi^V_r},\qquad g^\theta_r=\Pra\Ra\,\alpha^2\,\overline{\phi^\theta_r}.
\end{equation}
At finite $N$, $(\phi^V_r)'(1)=0$, and the $u$-term integrates by parts onto $V$:
\begin{equation}\label{eq:rc-g}
\begin{gathered}
g^u_r=0,\qquad g^V_r=-\sum_{j\le N}K_{r+1,j}\,\overline{\Lap V_j},\\
g^\theta_r=\Pra\Ra\,\alpha^2\sum_{j\le N}K_{r+1,j}\,\overline{\theta_j}.
\end{gathered}
\end{equation}\vspace{-\baselineskip}\\
\bottomrule
\end{tabular}}
\end{table}

\subsection{Why three wall inputs, and what they achieve}

\begin{theorem}[Three wall inputs control every mode of the RB layer]\label{thm:three}
Fix $\alpha\ne0$, $\Pra$, $\Ra$. The states reachable from zero by the three wall inputs are dense in $\cX$, and every finite spectral block of $\cA$ is controllable from the three wall functionals \eqref{eq:functionals}.
\end{theorem}

Theorem~\ref{thm:three} is proved in Section~\ref{sec:ref}. Its proof uses all three wall functionals: under $(b_1,b_2)$ alone one Cauchy datum at $y=1$, $\zeta'''(1)$, is left free. The two inputs that carry the Fredholm gains are the tangential velocity $U_c$ and the wall temperature $T_c$, the two whose wall coefficients grow linearly in the joint index by Proposition~\ref{prop:joint}(iii); wall heat flux in place of the wall temperature would give coefficients bounded in the index, and the coefficients of $V_c$ are smaller by the factor $\tanh(\alpha/2)$ on the modes symmetric about the mid-plane. So $V_c$ is driven by a finite-rank feedback: the pole placement of Proposition~\ref{prop:low}, which makes the low block controllable at every plant, and the Fredholm stage on finitely many middle clusters, which removes every rank condition on the design.

\section{Simulations}\label{sec:sim}

Or, Cortelezzi and Speyer \cite{OCS01} gave the first model-based feedback stabilization of the layer: linear-quadratic-Gaussian and $\mathcal H_\infty$ controllers that actuate only the temperature of the lower plate, estimate the state from measurements, and hold the state of rest to $\Ra\approx25{,}000$, $14.8$ times the critical value $1708$. Their result shows how hard it is to push stabilization an order of magnitude above the critical Rayleigh number. We simulate the recovery of the equilibrium from saturated convection at $\Ra=50{,}000$, $29$ times the critical value, and at $\Ra=10^6$ the release of saturated convection when the feedback is switched off.

The layer is simulated at $\Pra=7$, with period $L=4.032$ in $x$, by a Fourier--Chebyshev method, Crank--Nicolson for diffusion and Adams--Bashforth for advection and buoyancy, with the time step set by the flow and wall velocities.

\paragraph*{Recovery from saturated convection at $\Ra=50{,}000$} The equilibrium has $15$ unstable eigenvalues, in the wavenumbers $m=1,\dots,9$, listed in Table~\ref{tab:eig5}; the uncontrolled layer settles into steady convection rolls with three periods in $x$, Figure~\ref{fig:rec}, top. The grid has $128$ Fourier modes in $x$ and $96$ Chebyshev modes in $y$. The gains are recomputed every $2\times10^{-4}$ time units by steps~1, 2 and~4 of Table~\ref{tab:static}, with the plant linearized about the current mean temperature profile of the layer and each eigenvalue $\lambda_j$ to the right of $-\sigma$ placed at $-\lambda_j-2\sigma$. The shift $\sigma$ is scheduled on the kinetic energy $E$: it starts at $15$ and is raised in steps to $80$ as $E$ falls from $70\%$ to $10\%$ of its saturated value. On the wavenumbers $m=1,2,4,5$ adjacent to the rolls, the shift is $\sigma/4$ until the decrease of $E$ stalls, after which the full $\sigma$ is applied at every wavenumber. In the gain computation the tangential wall velocity is weighted $30$ times less, relative to the wall temperature, than at $\Ra=10^6$. Feedback switched on in saturated convection returns the layer to the equilibrium, Figure~\ref{fig:rec}: $E$ is below $10^{-7}$ of its saturated value by $t=0.24$. The tangential wall velocity stays within twice the largest velocity of the saturated flow, except for a peak of four times at $t=0.01$, and the wall temperature within $0.7$ of the temperature difference across the layer.

\begin{figure}[p]
\centering
\includegraphics[width=\columnwidth]{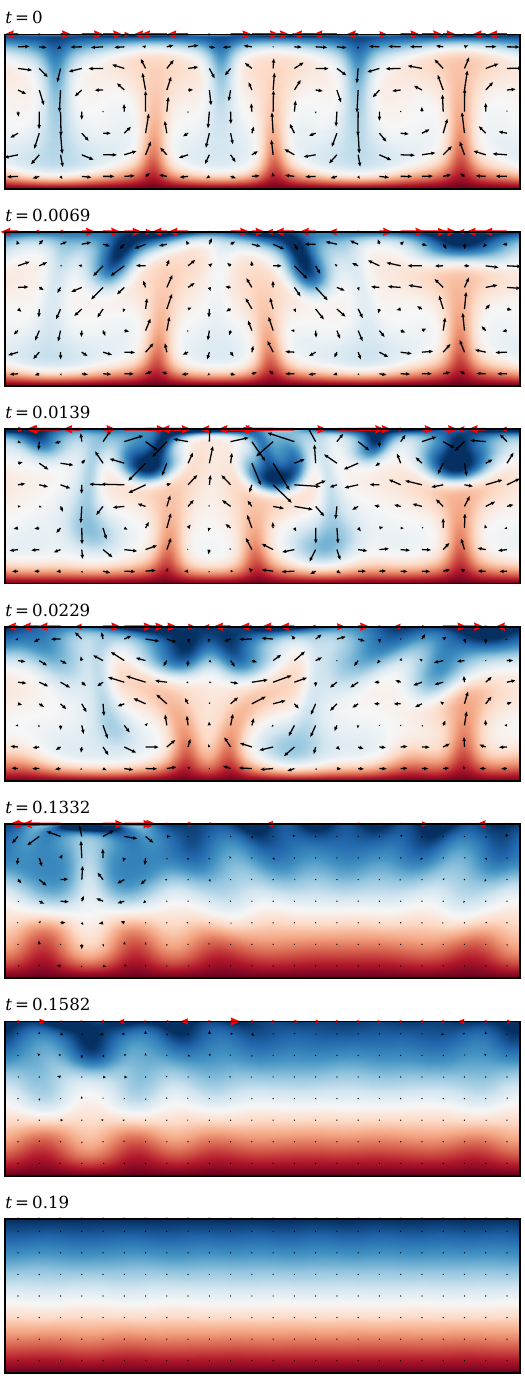}
\caption{Recovery of the equilibrium from saturated convection at $\Ra=50{,}000$, feedback switched on at $t=0$. Temperature (color, from $0$ at the upper plate to $1$ at the lower plate), velocity (black arrows) and wall velocity at $y=1$ (red arrows), all arrows to one scale. The vertical scale of the panels is stretched.}
\label{fig:rec}
\end{figure}

\begin{table}[t]
\caption{The $15$ unstable eigenvalues of the linearized RB layer at $\Pra=7$, $\Ra=50{,}000$, $L=4.032$.}\label{tab:eig5}
\centering
\begin{tabular}{@{}rrl@{}}
\toprule
$m$ & $\alpha_m$ & unstable eigenvalues\\
\midrule
1 & 1.56 & 131 \\
2 & 3.12 & 271, 43 \\
3 & 4.67 & 325, 89 \\
4 & 6.23 & 321, 110 \\
5 & 7.79 & 284, 106 \\
6 & 9.35 & 228, 81 \\
7 & 10.91 & 162, 41 \\
8 & 12.47 & 90 \\
9 & 14.02 & 16 \\
\bottomrule
\end{tabular}
\end{table}

\paragraph*{Release of convection at $\Ra=1{,}000{,}000$} The controlled runs use $192$ Fourier modes in $x$ and $96$ Chebyshev modes in $y$, time step $10^{-5}$, and the uncontrolled runs $256\times128$ modes, which the turbulent flow requires. The equilibrium of the RB layer has $85$ unstable eigenvalues, listed in Table~\ref{tab:eig}: in the wavenumbers $m=1,\dots,20$, at most five per wavenumber, $29$ of them above $1000$, the largest $2137$ at $m=5$.

The gains are computed by steps~1, 2 and~4 of Table~\ref{tab:static} at the wavenumbers $m=1,\dots,23$, with the eigenpairs and wall rows taken from the discretized linearization. The low block consists of the eigenvalues to the right of $-100$, and step~4 places each of them, $\lambda_j$, at $-\lambda_j-200$ in place of $\nu_j$, with the three wall inputs weighted in the gain computation; the eigenvalues to the left of $-100$ are left in place, $F=0$ in step~8. The bottom plate temperature is not used by the feedback.

From the equilibrium perturbed by a temperature disturbance of amplitude $10^{-6}$ on the wavenumbers $m=1,\dots,40$, the kinetic energy decays at the rate $200$ after an initial transient, down to round-off level in longer runs. When the feedback is switched off, the round-off perturbations grow at the rates of the $85$ unstable eigenvalues. The layer erupts within $0.014$ time units into the state of Figure~\ref{fig:sim}(b), with finer structure than the saturated convection of Figure~\ref{fig:sim}(a), which it approaches within a further $0.01$ time units. The controlled layer is not shown: its temperature is the equilibrium profile.

\begin{figure*}[t]
\centering
\includegraphics[width=.8\textwidth]{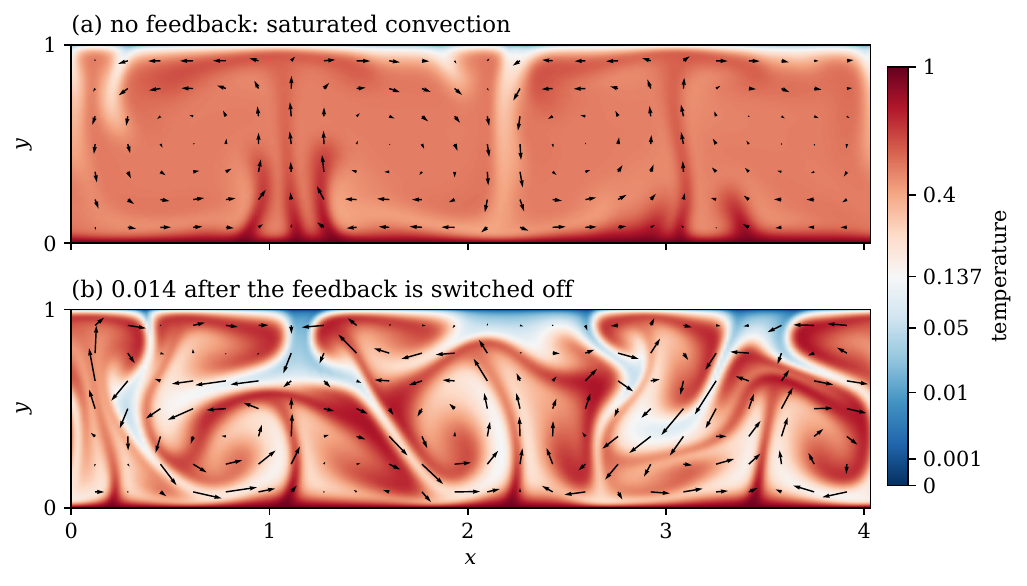}
\caption{Temperature and velocity of the RB layer at $\Ra=10^6$: (a) without feedback, saturated convection; (b) $0.014$ time units after the feedback is switched off. The arrows are the velocity vectors, drawn to the same scale in both panels. The color scale is $c=a\ln(1+T/\tau)/\ln(1+1/\tau)+(1-a)T$, with $\tau=0.001$, $a=0.63$, and is the same in both panels.}
\label{fig:sim}
\end{figure*}

\begin{table}[t]
\caption{The $85$ unstable eigenvalues of the linearized RB layer at $\Pra=7$, $\Ra=10^6$, $L=4.032$: $4$ of them above $2000$ and $25$ more above $1000$. Wavenumbers $\alpha_m=2\pi m/L$.}\label{tab:eig}
\centering
\begin{tabular}{@{}rrl@{}}
\toprule
$m$ & $\alpha_m$ & unstable eigenvalues\\
\midrule
1 & 1.56 & 960, 337, 58 \\
2 & 3.12 & 1652, 817, 340, 47 \\
3 & 4.67 & 1980, 1185, 614, 220 \\
4 & 6.23 & 2110, 1428, 836, 383, 53 \\
5 & 7.79 & 2137, 1567, 997, 519, 148 \\
6 & 9.35 & 2103, 1629, 1099, 620, 227 \\
7 & 10.91 & 2031, 1632, 1149, 684, 283 \\
8 & 12.47 & 1933, 1593, 1157, 715, 317 \\
9 & 14.02 & 1813, 1521, 1129, 715, 328 \\
10 & 15.58 & 1679, 1424, 1073, 688, 317 \\
11 & 17.14 & 1532, 1310, 994, 639, 287 \\
12 & 18.70 & 1377, 1181, 897, 571, 239 \\
13 & 20.26 & 1215, 1041, 786, 487, 176 \\
14 & 21.82 & 1048, 894, 664, 390, 100 \\
15 & 23.37 & 878, 741, 534, 283, 12 \\
16 & 24.93 & 707, 584, 397, 167 \\
17 & 26.49 & 535, 425, 255, 44 \\
18 & 28.05 & 364, 264, 110 \\
19 & 29.61 & 193, 102 \\
20 & 31.17 & 22 \\
\bottomrule
\end{tabular}
\end{table}

\section{Clusterwise Fredholm design}\label{sec:proof}

\subsection{How to read this section}

What this section works toward is Theorem~\ref{thm:main}, at the end of Section~\ref{sec:fredholm}: at a single wavenumber, and for a decay rate $\omega$ fixed in advance, the three wall inputs have gain functions in $L^2(0,1)$ for which the closed-loop spectrum is the uncoupled spectrum shifted past $-\omega$, and the velocity and the temperature decay at the rate $\omega$.

The design that produces those gains is backstepping, driven by the target \eqref{eq:target}: a bounded, boundedly invertible transformation $T$ of the state is required to map the closed loop onto it, so that the two have the same decay. The target is the plant's two subsystems with their coupling deleted and their spectra shifted by $q_1$ and $q_2$, which the designer sets, so for any rate $\omega$ specified in advance the whole state decays at that rate. The rate is not obtained by moving a few unstable modes: the infinitely many eigenvalues of both subsystems are shifted at once. The plant operator, the wall functionals through which the three inputs act on it, and the spectra of the two uncoupled subsystems are in Section~\ref{sec:ref}.

In a Volterra design the kernel of $T$ solves a partial differential equation and the feedback is that kernel evaluated at the wall. Here the transformation is built from its action on a basis, in Section~\ref{sec:fredholm}, since the plant admits no Volterra transformation onto this target. The target eigenvectors are the uncoupled eigenvectors $\tilde e_\nu$, so the intertwining forces $T^{-1}\tilde e_\nu$ to be the closed-loop eigenvector at the target eigenvalue $\nu$. Such a state is the response of the plant at the complex frequency $\nu$ to a wall input vector $c_\nu$, that is, the solution $\psi_\nu$ of $(\nu-\cA)\psi_\nu=Bc_\nu$, and it is a closed-loop eigenvector once the feedback returns $c_\nu$ when applied to $\psi_\nu$. The pair $(\psi_\nu,c_\nu)$, one for every target eigenvalue, takes the place of the kernel equations; the feedback is read off from the whole family, which is where the gain functions of Theorem~\ref{thm:main} come from.

What decides whether $T$ is bounded with a bounded inverse is how far the states $\psi_\nu$ are from the eigenvectors of the plant, since $T^{-1}$ carries the orthonormal basis $\{\tilde e_\nu\}$ to $\{\psi_\nu\}$. Shifting one eigenvalue changes its eigenvector little when the neighboring eigenvalues are far away, and $\psi_\nu$ is then close to the plant eigenvector it comes from. The spacing of the RB spectrum, however, is not uniform: a velocity eigenvalue falls close to a temperature eigenvalue at infinitely many indices, at every wavenumber, and with a single input the two states prescribed there would be nearly parallel. The spectrum is therefore cut into clusters in Section~\ref{sec:clusters}, each cluster a single eigenvalue or a close pair from the two subsystems, and each is assigned as a unit by the two inputs $U_c$ and $T_c$, which keeps the two states apart. How large the wall coefficients are on a cluster, which is what makes the family $\{\psi_\nu\}$ summable, and the pole placement of the finitely many lowest eigenvalues, whose spacing is of the order of the shifts $q_1,q_2$, are in Section~\ref{sec:Bm}. Proofs of supporting results are in Appendix~\ref{app:proofs}.

\paragraph*{An LTI ODE analogy of the section's results} In the eigenbasis of the plant operator, which Proposition~\ref{prop:plant} provides, the plant at one wavenumber is a diagonal system with three input columns,
\begin{equation}\label{eq:diag}
\dot x_j=\lambda_jx_j+b_j\cdot u,\qquad u=(V_c,U_c,T_c),
\end{equation}
with $\lambda_1\ge\lambda_2\ge\cdots\to-\infty$ real and $b_j\in\mathbb C^3$ the wall coefficients of the $j$-th mode. Everything up to Proposition~\ref{prop:inv} concerns this pair.

The eigenvalues of the uncoupled operator are known in closed form, and its wall coefficients grow like $j$ (Proposition~\ref{prop:joint}). The coupling between the two subsystems is a bounded operator, so each plant eigenvalue sits within a fixed distance of an uncoupled one (Proposition~\ref{prop:cluster}). Since consecutive eigenvalues of each family are $O(j)$ apart, that distance resolves them one by one, except where a velocity eigenvalue and a temperature eigenvalue nearly collide; those two are kept together in a cluster and treated as a $2\times2$ block, for which the two inputs $U_c$ and $T_c$ give a uniformly invertible input matrix (Proposition~\ref{prop:Bm}, Corollary~\ref{cor:BG}).

Assigning the closed-loop eigenvalues then follows the finite-dimensional rule. If $A+BF$ is to have the eigenvalue $\nu$ with the eigenvector $\psi$, then $(\nu-A)\psi=Bc$ with $c=F\psi$, so
\begin{equation}\label{eq:assign}
\psi=\sum_j\frac{b_j\cdot c}{\nu-\lambda_j}\,X_j ,
\end{equation}
and the input direction $c$ selects the eigenvector. Here \eqref{eq:assign} is imposed at every target eigenvalue at once. Each $\nu$ is an uncoupled eigenvalue shifted by $q_1$ or $q_2$, and $|b_j|$ grows like $j$, so $|c|$ is of order $1/j$ and $\psi$ differs from the uncoupled eigenvector by $O(1/j)$, square-summably over the spectrum (Lemma~\ref{lem:eig}). That is what makes the transformation of these eigenvectors onto the uncoupled basis bounded with a bounded inverse, which is Proposition~\ref{prop:inv}. The finitely many lowest modes, whose spacing is of the order of the shifts, are placed beforehand by a finite gain matrix, as in a finite-dimensional plant (Proposition~\ref{prop:low}).

The section, then, in one sentence: the plant is diagonal with input columns growing like $j$; it is a bounded perturbation of an explicit uncoupled, resolved mode by mode except at the collisions, which are handled in pairs by two inputs; every closed-loop eigenvector is prescribed by \eqref{eq:assign}, and the growth of the input columns makes the prescribed family close enough to the uncoupled basis for the transformation to be invertible.

\subsection{Plant operator, wall functionals, and uncoupled spectra}\label{sec:ref}

With $V_c\ne0$ the velocity profile does not vanish at the upper wall and is not an element of $\cX$; the physical state space adds one direction to $\cX$.

\begin{proposition}[Plant is self-adjoint, and its inputs act through three wall functionals]\label{prop:plant}
Fix $\alpha\ne0$.
\begin{itemize}
\item[(i)] $\cA$ is self-adjoint on $\cX$ with compact resolvent. Its eigenvalues are real, of finite multiplicity, its eigenvectors form an orthonormal basis of $\cX$, and it commutes with the reflection $\mathcal P(V,\theta)(y)\coloneqq(V,\theta)(1-y)$, so every eigenvector can be taken of definite parity.
\item[(ii)] Let $(V,\theta)$ solve the resolvent problem \eqref{eq:plant} at $s\in\mathbb C$ with wall data $(V_c,U_c,T_c)$, and let $(e,\vartheta)$ be a normalized eigenvector of $\cA$ with eigenvalue $\lambda$. Then
\begin{equation}\label{eq:residue}
\begin{split}
(s-\lambda)\,\ip{(V,\theta)}{(e,\vartheta)}_\cX
&=\Pra\,e'''(1)\,V_c+i\alpha\Pra\,e''(1)\,U_c\\
&\quad-\Pra\Ra\,\alpha^2\,\vartheta'(1)\,T_c .
\end{split}
\end{equation}
\item[(iii)] Let $\hat\cX\coloneqq\{(V,\theta)\in H^1(0,1)\times L^2(0,1):V(0)=0\}$ with the norm of Definition~\ref{def:ops}, and $h_\alpha(y)\coloneqq\sinh(\alpha y)/\sinh\alpha$. Then $\hat\cX=\cX\oplus\spn\{(h_\alpha,0)\}$ orthogonally, with $P_0$ the orthogonal projection on $\cX$, $\|(h_\alpha,0)\|^2=\alpha\coth\alpha$, and $\|u\|^2+\|V\|^2+\Pra\Ra\|\theta\|^2=\|X\|^2/\alpha^2$ on $\hat\cX$. For a classical solution $X=(V,\theta)$ of \eqref{eq:plant} with wall data $V_c,U_c,T_c$, the component $\tilde X\coloneqq X-V_c(h_\alpha,0)\in\cX$ satisfies $\ip{\tilde X}{X_j}_\cX=\ip{X}{X_j}_\cX$ and
\begin{equation}\label{eq:modal-ode}
\begin{split}
\frac{\dd}{\dd t}\ip{\tilde X}{X_j}_\cX&=\lambda_j\ip{\tilde X}{X_j}_\cX+b_0[X_j]V_c\\
&\quad+b_1[X_j]U_c+b_2[X_j]T_c
\end{split}
\end{equation}
for every eigenvector $X_j$ of $\cA$, $\lambda_j$ its eigenvalue.
\end{itemize}
\end{proposition}

Throughout, the eigenvalues $\lambda_1\ge\lambda_2\ge\cdots$ of $\cA$ are listed with multiplicity with a real orthonormal eigenbasis $(X_j)_{j\ge1}$, which exists by (i) of Proposition~\ref{prop:plant} because $\cA$ has real coefficients, and $b_i[\cdot]$, $i=0,1,2$, are the wall functionals \eqref{eq:functionals}. By (ii), the wall rows of step~2 of Table~\ref{tab:static} are the coefficients of the three inputs, and the functionals \eqref{eq:functionals} are conjugate to $V_c$, $U_c$, $T_c$.

\begin{definition}[Joint index]\label{def:joint}
$\{\varsigma_j\}_{j\ge1}$ is the sequence of the eigenvalues $\{\lambda^S_n\}_{n\ge1}$ of $S_\alpha$ and $\{\delta_m\}_{m\ge1}$ of $\Lap$ on $\mathcal D_\theta$, listed in decreasing order; $j(\lambda^S_n)$, $j(\delta_m)$ are the positions.
\end{definition}

\begin{proposition}[Uncoupled spectra, and linear growth of their wall coefficients in the joint index]\label{prop:joint}
Fix $\alpha\ne0$.
\begin{itemize}
\item[(i)] $S_\alpha$ is self-adjoint and negative on $\Ha$ with simple eigenpairs $(e_n,\lambda^S_n)_{n\ge1}$,
\begin{equation}\label{eq:lam}
\lambda^S_n=-\Pra\big(\sigma_n^2+\alpha^2\big),
\end{equation}
$\sigma_n$ the increasing positive roots of
\begin{subequations}\label{eq:rootsys}
\begin{align}
\sigma\sin\tfrac\sigma2+\alpha\tanh\tfrac\alpha2\cos\tfrac\sigma2&=0\qquad(\text{symmetric}),\label{eq:roots}\\
\sigma\cos\tfrac\sigma2-\alpha\coth\tfrac\alpha2\sin\tfrac\sigma2&=0\qquad(\text{antisymmetric}),\label{eq:roots2}
\end{align}
\end{subequations}
\begin{equation}\label{eq:sigma}
\begin{gathered}
\sigma_n=(n+1)\pi-\frac{c_n}{\sigma_n}+O(\sigma_n^{-3}),\\
c_n=\begin{cases}2\alpha\tanh(\alpha/2),&n+1\ \text{even},\\2\alpha\coth(\alpha/2),&n+1\ \text{odd},\end{cases}
\end{gathered}
\end{equation}
the symmetric roots being those with $n+1$ even. $\Lap$ on $\mathcal D_\theta$ is self-adjoint and negative on $L^2$ with eigenpairs
\begin{equation}\label{eq:delta}
\begin{gathered}
\theta_m=\frac{\sqrt2\sin m\pi y}{\alpha\sqrt{\Pra\Ra}},\\
\delta_m=-\big(m^2\pi^2+\alpha^2\big),\qquad m\ge1,
\end{gathered}
\end{equation}
so that $\|(0,\theta_m)\|_\cX=1$; $\theta_m$ is symmetric about $y=\tfrac12$ for $m$ odd and antisymmetric for $m$ even. With $e_n$ normalized in $\Ha$,
\begin{subequations}\label{eq:trsys}
\begin{gather}
|e_n''(1)|=\sqrt2\,\sigma_n\big(1+O(\sigma_n^{-2})\big),\label{eq:etr}\\
|e_n'''(1)|=\tfrac{\sqrt2}{2}\,c_n\,\sigma_n\big(1+O(\sigma_n^{-2})\big),\label{eq:etr3}\\
|\theta_m'(1)|=\frac{\sqrt2\,m\pi}{\alpha\sqrt{\Pra\Ra}} .\label{eq:ttr}
\end{gather}
\end{subequations}
\item[(ii)] The joint index satisfies
\begin{subequations}\label{eq:jidxsys}
\begin{gather}
j(\lambda^S_n)=\big(1+\Pra^{1/2}\big)n+O(1),\label{eq:jidx}\\
j(\delta_m)=\big(1+\Pra^{-1/2}\big)m+O(1),\label{eq:jidx2}
\end{gather}
\end{subequations}
\item[(iii)] The wall coefficients of the uncoupled modes satisfy
\begin{subequations}\label{eq:bsys}
\begin{gather}
|b_1[e_n]|=\sqrt2\,\alpha\Pra\,\sigma_n\big(1+O(n^{-2})\big)\asymp j,\label{eq:bgrow}\\
|b_2[\theta_m]|=\sqrt2\,\pi\,\alpha\sqrt{\Pra\Ra}\;m\asymp j,\label{eq:bgrow2}\\
|b_0[e_n]|=\tfrac{\sqrt2}{2}\,\Pra\,c_n\,\sigma_n\big(1+O(n^{-2})\big)\asymp j .\label{eq:b0grow}
\end{gather}
\end{subequations}
Under wall heat flux, $-k_{\rm th}\theta_y(1)=q_c$, the uncoupled thermal eigenvectors are $\sqrt2\sin(m+\tfrac12)\pi y$, $m\ge0$, and the wall coefficient of $q_c$ is $\sqrt2\,\alpha\sqrt{\Pra\Ra}\,/k_{\rm th}$, bounded in $m$.
\end{itemize}
\end{proposition}

\begin{proof}[Proof of Theorem~\ref{thm:three}]
No eigenvector $(\zeta,\psi)\ne0$ of $\cA$ is annihilated by the three wall functionals. The eigenproblem of $\cA$ is the sixth-order system
\begin{equation}\label{eq:eig}
\Pra\Lap^2\zeta-\lambda\Lap\zeta=\Pra\Ra\,\alpha^2\psi,\qquad\Lap\psi-\lambda\psi=-\zeta,
\end{equation}
with Cauchy data $\zeta,\zeta',\zeta'',\zeta''',\psi,\psi'$ at $y=1$. The domain gives $\zeta(1)=\zeta'(1)=\psi(1)=0$; the three functionals give $\zeta''(1)=\zeta'''(1)=\psi'(1)=0$; uniqueness for linear ordinary differential systems gives $(\zeta,\psi)\equiv0$. Density follows by Hahn--Banach from \eqref{eq:residue}, and controllability of finite blocks by the Hautus test.
\end{proof}

\subsection{Clustering of the plant spectrum}\label{sec:clusters}

The uncoupled operator $\cS=\operatorname{diag}(S_\alpha,\Lap)$ of Definition~\ref{def:ops} has the eigenvectors $\tilde e_\nu\in\{(e_n,0)\}\cup\{(0,\theta_m)\}$ and the eigenvalues $\lambda^S_n$, $\delta_m$, listed jointly as $\varsigma_1\ge\varsigma_2\ge\cdots$ (Definition~\ref{def:joint}). The target \eqref{eq:target} is $\cS$ with each of its two subsystems shifted, and the plant is $\cS$ plus the coupling.

\begin{definition}[Clusters]\label{def:clusters}
For $\epsilon\in(0,\epsilon_0)$, $\epsilon_0\coloneqq\tfrac12\pi\min(1,\sqrt{\Pra})$, join two consecutive entries $a,b$ of the joint list $\varsigma_1\ge\varsigma_2\ge\cdots$ when $|a-b|\le\epsilon\min(\sqrt{|a|},\sqrt{|b|})$; the clusters are the connected components. For a cluster $G$ of index $j$, $n_G$ is the number of its elements, $D_G$ is its diameter, $g_G\coloneqq\dist(G,\sigma(\cS)\setminus G)$, $\Gamma_G$ is the circle centered at the midpoint of $G$ with radius $r_G\coloneqq\tfrac12D_G+\tfrac12g_G$, and $Q^0_G$, $Q_G$ are the spectral projections of $\cS$ and $\cA$ for the interior of $\Gamma_G$.
\end{definition}

\begin{proposition}[Plant spectrum splits into clusters close to the uncoupled clusters]\label{prop:cluster}
Fix $\alpha\ne0$ and $\epsilon\in(0,\epsilon_0)$.
\begin{itemize}
\item[(i)] $\mathcal K\coloneqq\cA-\cS$, $\mathcal K(V,\theta)=(-\Pra\Ra\,\alpha^2\Lap^{-1}\theta,\ V)$, is a bounded self-adjoint operator on $\cX$ with
\begin{equation}\label{eq:kappa0}
\|\mathcal K\|\le\kappa_0\coloneqq\frac{\alpha\sqrt{\Pra\Ra}}{\sqrt{\pi^2+\alpha^2}} ,
\end{equation}
and $|\lambda_j-\varsigma_j|\le\kappa_0$ for every $j$.
\item[(ii)] There is $j_0$ such that every cluster of index $j\ge j_0$ has at most two elements, of different families, with diameter $D_G\le\epsilon\sqrt{|\varsigma_j|}$ and distance $g_G\coloneqq\dist(G,\sigma(\cS)\setminus G)>\epsilon\sqrt{|\varsigma_j|}(1-o(1))$, $\sqrt{|\varsigma_j|}\asymp j$, and
\begin{equation}\label{eq:rG}
r_G\asymp j,\qquad\dist(\Gamma_G,\sigma(\cS))\ge\tfrac12g_G\asymp j .
\end{equation}
\item[(iii)] There is $j_1\ge j_0$ such that for $j\ge j_1$ the circle $\Gamma_G$ encloses exactly $n_G$ eigenvalues of $\cA$, counted with multiplicity, and
\begin{equation}\label{eq:QQ}
\|Q_G-Q^0_G\|\le\frac{C_G\kappa_0}{j}<\tfrac12,\qquad\sum_G\|Q_G-Q^0_G\|^2<\infty ,
\end{equation}
with $C_G$ depending on $\epsilon,\alpha,\Pra$. In particular $\sum_\nu\|(I-Q_{G(\nu)})\tilde e_\nu\|^2<\infty$ over the high uncoupled eigenvectors.
\end{itemize}
\end{proposition}

\begin{lemma}[Eigenvectors in one cluster have opposite parity at unit Prandtl number]\label{lem:parity}
There is $n_0(\alpha)$ such that for $n\ge n_0$, if $\sigma_n$ lies within $\pi/2$ of $m\pi$ then $e_n$ is symmetric about $y=\tfrac12$ if and only if $m$ is even, $\theta_m$ is symmetric if and only if $m$ is odd, and
\begin{equation}\label{eq:zero}
\ip{\Lap^{-1}\theta_m}{e_n}_\alpha=\ip{e_n}{\theta_m}=0 .
\end{equation}
At $\Pra=1$ every high cluster of Definition~\ref{def:clusters} is a pair $\{\lambda^S_{m-1},\delta_m\}$ with
\begin{equation}\label{eq:gap1}
\lambda^S_{m-1}-\delta_m=m^2\pi^2-\sigma_{m-1}^2=2c_{m-1}+O(m^{-2}),
\end{equation}
bounded uniformly in $m$, and the two eigenvectors of the pair are of opposite parity. At other Prandtl numbers a pair need not have opposite parity: at $\Pra=4$, $\lambda^S_n-\delta_{2(n+1)}=8c_n-3\alpha^2+O(n^{-2})$ for even $n$, and both eigenvectors are antisymmetric.
\end{lemma}

\begin{proof}
By \eqref{eq:sigma}, for $n$ large the root $\sigma_n$ nearest $m\pi$ is on the symmetric branch when $m$ is even (at small $n$ this fails: at $\alpha=100$ the first symmetric root is within $0.07$ of $\pi$); $\sin m\pi(\tfrac12+z)=\pm\sin m\pi z$ for $m$ even and $\pm\cos m\pi z$ for $m$ odd. Functions of opposite parity are orthogonal in $L^2$ and in $\Ha$, and $\Lap^{-1}$ preserves parity. \eqref{eq:gap1} is \eqref{eq:sigma} multiplied by $m\pi+\sigma_{m-1}$; a bounded gap is below the threshold $\epsilon\sqrt{|\varsigma_j|}$ of Proposition~\ref{prop:cluster}(ii) for $j$ large, so the two lie in one cluster. The $\Pra=4$ example is $4\sigma_n^2-4(n+1)^2\pi^2$ with \eqref{eq:sigma}.
\end{proof}

\subsection{Wall coefficients and low-block pole placement}\label{sec:Bm}

Let $G$ be a high cluster with plant eigenvalues $\{\lambda_j\}_{j\in G}$. The plant wall matrix is
\begin{equation}\label{eq:PhiG}
\Phi_G\coloneqq\big(b_1[X_j],\ b_2[X_j]\big)_{j\in G}\in\mathbb C^{n_G\times2}.
\end{equation}
Since $\{X_j\}_{j\in G}$ is an orthonormal basis of $Q_G\cX$, the singular values of $\Phi_G$ are those of the map $\Phi:Y\mapsto(b_1[Y],b_2[Y])$ from $Q_G\cX$ to $\mathbb C^2$. On the uncoupled subspace $Q^0_G\cX$, in the uncoupled basis, $\Phi$ has the matrix $\operatorname{diag}(b_1[e_n],b_2[\theta_m])$ for a pair and the row $(b_1[e_n],0)$ or $(0,b_2[\theta_m])$ for a singleton, because $b_1$ acts on the velocity component only and $b_2$ on the thermal component only; both nonzero entries are of order $j$ by \eqref{eq:bgrow}.

\begin{proposition}[Wall coefficients of the plant modes grow linearly on every cluster]\label{prop:Bm}
Fix $\alpha\ne0$ and $\epsilon\in(0,\epsilon_0)$.
\begin{itemize}
\item[(i)] For $s\notin\sigma(\cA)$ and $i\in\{1,2\}$, the resolvent problem \eqref{eq:plant} at $s$ with $U_c=1$, $V_c=T_c=0$ for $i=1$, and with $T_c=1$, $V_c=U_c=0$ for $i=2$, has a unique solution $Y_{i,s}\in\cX$, and
\begin{equation}\label{eq:dual}
\begin{gathered}
\big|b_i\big[(z-\cA)^{-1}f\big]\big|=\big|\ip{Y_{i,\bar z}}{f}\big|,\\
f\in\cX,\quad z\notin\sigma(\cA).
\end{gathered}
\end{equation}
The same holds for $\cS$ in place of $\cA$, with solutions $Y^0_{i,s}$. There is $C_Y$, depending on $\alpha,\Pra,\Ra,\epsilon$, such that for every high cluster $G$ of index $j\ge j_1$ and every $s\in\Gamma_G$
\begin{equation}\label{eq:respbound}
\|Y_{i,s}\|\le2\|Y^0_{i,s}\|\le C_Y,\qquad i=1,2 .
\end{equation}
\item[(ii)] There are $c_1,C_1>0$ and $j_2\ge j_1$, depending on $\alpha,\Pra,\Ra,\epsilon$, such that $|b_j|\le C_1j$ for every $j\ge1$.
\item[(iii)] For every high cluster $G$ of index $j\ge j_2$, $\sigma_{\min}(\Phi_G)\ge2c_1j$ if $n_G=2$ and $|\Phi_G|\ge2c_1j$ if $n_G=1$.
\end{itemize}
\end{proposition}

\begin{proof}[Proof of Proposition~\ref{prop:Bm}(ii)--(iii)]
(ii), inputs $U_c$ and $T_c$. Let $G$ be a high cluster of index $j\ge j_1$ and $Y\in Q_G\cX$. The integral $Q_GY=\frac1{2\pi i}\oint_{\Gamma_G}(z-\cA)^{-1}Y\,\dd z$ converges in the graph norm of $\cA$, on which $b_1$ and $b_2$ are continuous, so by \eqref{eq:dual} and \eqref{eq:respbound}
\begin{equation}\label{eq:bQ}
\begin{split}
|b_i[Y]|&\le r_G\sup_{z\in\Gamma_G}\|Y_{i,\bar z}\|\,\|Y\|\\
&\le C_Y\,r_G\,\|Y\|\le Cj\,\|Y\| ,\qquad i=1,2 .
\end{split}
\end{equation}
Every plant eigenvector $X_j$ with $j\ge j_1$ lies in $Q_G\cX$ for the cluster of $\varsigma_j$, because $|\lambda_j-\varsigma_j|\le\kappa_0<\tfrac12g_G$ by Proposition~\ref{prop:cluster}(i). The finitely many $X_j$ with $j<j_1$ are smooth.

(ii), input $V_c$. The identity \eqref{eq:b0id} extends by linearity to every $Y\in Q_G\cX$. The reflection $\mathcal P$ of Proposition~\ref{prop:plant}(i) is unitary and commutes with $Q_G$, and the second derivative at $y=0$ of the velocity component of $Y$ is that at $y=1$ of $\mathcal PY$; so \eqref{eq:bQ} bounds both second derivatives by $Cj\|Y\|$, and $|b_0[Y]|\le Cj\|Y\|$.

(iii). By \eqref{eq:QQint}, for $Y\in Q_G\cX$,
\begin{equation}\label{eq:bQQ}
\begin{split}
&b_i\big[(Q_G-Q^0_G)Y\big]\\
&\quad=\frac1{2\pi i}\oint_{\Gamma_G}b_i\big[(z-\cA)^{-1}\mathcal K(z-\cS)^{-1}Y\big]\,\dd z ,
\end{split}
\end{equation}
and by \eqref{eq:dual}, \eqref{eq:respbound}, \eqref{eq:kappa0} and $\|(z-\cS)^{-1}\|\le2/g_G$ on $\Gamma_G$, its modulus is at most $2C_Y\kappa_0(r_G/g_G)\|Y\|\le C_2\|Y\|$, since $r_G/g_G\le2$ for $j$ large by \eqref{eq:rG} and Proposition~\ref{prop:cluster}(ii). By \eqref{eq:bgrow} there is $c_0>0$ with $|\Phi Z|\ge c_0j\|Z\|$ for $Z\in Q^0_G\cX$. For a unit $Y\in Q_G\cX$, $\|Q^0_GY\|\ge1-C_G\kappa_0/j$ by \eqref{eq:QQ}, so
\begin{equation}\label{eq:Philow}
\begin{split}
|\Phi Y|&\ge|\Phi Q^0_GY|-\sqrt2\,C_2\\
&\ge c_0j-C_Gc_0\kappa_0-\sqrt2\,C_2\ge\tfrac12c_0j
\end{split}
\end{equation}
for $j\ge j_2$. Since $\dim Q_G\cX=n_G$ and $\{X_j\}_{j\in G}$ is orthonormal, this is (iii) with $c_1\coloneqq c_0/4$.
\end{proof}

By Proposition~\ref{prop:plant}(iii), the wall inputs act on the modal coordinates through \eqref{eq:modal-ode}, so that in the eigenbasis the input operator $B=[B_{V_c}\ B_{U_c}\ B_{T_c}]$ has the rows $b_j\coloneqq(b_0[X_j],b_1[X_j],b_2[X_j])$. By Proposition~\ref{prop:Bm}(ii), $|b_j|\le C_1j$, so $B$ maps $\mathbb C^3$ into the extrapolation space $\cX_{-1}$, and for a bounded $\kappa:\cX\to\mathbb C^3$ the operator $\cA+B\kappa$ is defined on $\{x\in\cX:\cA_{-1}x+B\kappa x\in\cX\}$, with $\cA_{-1}$ the extrapolated operator on $\cX_{-1}$.

\begin{definition}[Low block and pole-placement pre-feedback]\label{def:P}
$P$ is the orthogonal projection onto $\operatorname{span}\{X_1,\dots,X_d\}$, where $d$ is chosen so that this span is a union of eigenspaces of $\cA$; $A_L\coloneqq\operatorname{diag}(\lambda_1,\dots,\lambda_d)$, $B_L\in\mathbb C^{d\times3}$ with rows $b_1,\dots,b_d$. $K_L\in\mathbb C^{3\times d}$ is a gain for which $A_L+B_LK_L$ has prescribed distinct eigenvalues $\mu_1,\dots,\mu_d$ in $\rho(\cA)$, with eigenvectors $x_j$ and input vectors $c_j\coloneqq K_Lx_j$. The pole-placement pre-feedback is $\kappa\coloneqq K_L\big(\ip{\cdot}{X_j}_\cX\big)_{j\le d}$, and the pre-compensated operator is $\cA'\coloneqq\cA+B\kappa$.
\end{definition}

\begin{proposition}[Pole placement on low block leaves tail unchanged]\label{prop:low}
Let $d$ and $\mu_1,\dots,\mu_d$ be as in Definition~\ref{def:P}.
\begin{itemize}
\item[(i)] $(A_L,B_L)$ is controllable, so a gain $K_L$ as in Definition~\ref{def:P} exists for every choice of distinct $\mu_1,\dots,\mu_d$. One such gain, valid whenever the matrix $X_L\coloneqq[x_1\cdots x_d]$ below is invertible, is $K_L\coloneqq[c_1\cdots c_d]X_L^{-1}$ with
\begin{subequations}\label{eq:Kgain}
\begin{gather}
c_j\coloneqq\frac{\overline{b_j}}{b_j\cdot\overline{b_j}}\,(\mu_j-\lambda_j),\qquad j\le d,\label{eq:cj}\\
(x_j)_i\coloneqq\begin{cases}1,&i=j,\\ \dfrac{b_i\cdot c_j}{\mu_j-\lambda_i},&i\ne j,\ i\le d.\end{cases}\label{eq:xj}
\end{gather}
\end{subequations}
\item[(ii)] With $\mu_1,\dots,\mu_d$ distinct in $\rho(\cA)$, $\cA'$ has compact resolvent,
\begin{equation}\label{eq:specAp}
\sigma(\cA')=\{\mu_1,\dots,\mu_d\}\cup\{\lambda_j\}_{j>d},
\end{equation}
the tail eigenvalues with their multiplicities in $\cA$, with eigenvectors
\begin{subequations}\label{eq:eigApsys}
\begin{gather}
X'_j\coloneqq X_j,\qquad j>d,\label{eq:eigAp}\\
\begin{aligned}X'_j\coloneqq{}&\sum_{i\le d}(x_j)_iX_i\\&+(\mu_j-\cA)^{-1}(I-P)Bc_j,\qquad j\le d,\end{aligned}\label{eq:eigAp2}
\end{gather}
\end{subequations}
and $\{X'_j\}_{j\ge1}$ is a Riesz basis of $\cX$ with dual basis $\{\tilde X'_j\}$, $\tilde X'_j=X_j$ for $j>d$ up to a correction in $\operatorname{span}\{X_i\}_{i\le d}$. The subspace $(I-P)\cX$ is invariant under $\cA'$, on which $\cA'=\cA$; $P\cX$ is invariant if and only if $(I-P)B\kappa=0$.
\end{itemize}
\end{proposition}

\begin{proof}
(i) Controllability: an eigenvector of $\cA$ annihilated by $b_0,b_1,b_2$ has, with its clamped and Dirichlet conditions, six vanishing Cauchy data at $y=1$ for the sixth-order system \eqref{eq:pV}--\eqref{eq:ptheta} and vanishes; this is the Hautus test on every vector of every eigenspace, so multiple eigenvalues are covered. Existence of $K_L$ is multi-input pole placement. For the displayed gain, row $i$ of $(\mu_j-A_L)x_j$ is $b_i\cdot c_j$, so $(\mu_j-A_L)x_j=B_Lc_j$, and $K_Lx_j=c_j$.
(ii) For $j>d$, $\kappa X_j=0$, so $X_j\in\mathcal D(\cA')$ and $\cA'X_j=\lambda_jX_j$; this gives the invariance and the tail of \eqref{eq:specAp}--\eqref{eq:eigAp}. For $j\le d$ write $X'_j=x+y$, $x\coloneqq\sum_{i\le d}(x_j)_iX_i$, $y\in(I-P)\cX$; then $\kappa X'_j=K_Lx_j=c_j$ and $\cA'X'_j=\mu_jX'_j$ splits into $(A_L+B_LK_L)x_j=\mu_jx_j$ and $\cA y+(I-P)Bc_j=\mu_jy$, solved by $y=(\mu_j-\cA)^{-1}(I-P)Bc_j\in\cX$. The resolvent: for $z\notin\sigma(\cA)\cup\{\mu_j\}$ and $f\in\cX$, the equation $(z-\cA')x=f$ with $x=Px+(I-P)x$ reads $(z-A_L-B_LK_L)Px=Pf$ and $(z-\cA)(I-P)x=(I-P)f+(I-P)BK_LPx$, solved in that order; the inverse is $(z-\cA)^{-1}$ plus an operator of rank at most $d$, hence compact. This also shows that \eqref{eq:specAp} is the whole spectrum. The Riesz property: $\{X_j\}_{j>d}$ is orthonormal, and the $d$ vectors \eqref{eq:eigAp2} have low coordinates $X_L$, which is invertible because the $x_j$ are eigenvectors of $A_L+B_LK_L$ for distinct eigenvalues; so $\{X'_j\}_{j\le d}$ together with $\{X_j\}_{j>d}$ is the image of the orthonormal basis under $I+(\text{rank }d)$, invertible. Invariance of $P\cX$: $\cA'P-P\cA'=(I-P)B\kappa$.
\end{proof}

Pole placement leaves the tail eigenvectors untouched, but the dual basis of the tail is changed by the tails of the low eigenvectors, and the input coefficients that enter the Fredholm stage are the ones in the dual basis: $\ip{Bc}{\tilde X'_j}_\cX$, not $b_j\cdot c$. This is used in Definition~\ref{def:psi}.

With $\BF\coloneqq[B_{U_c}\ B_{T_c}]$, the cluster input matrix of a high cluster $G$ is
\begin{equation}\label{eq:BG}
\mathbb B_G\coloneqq\big(\ip{B_{U_c}}{\tilde X'_j}_\cX,\ \ip{B_{T_c}}{\tilde X'_j}_\cX\big)_{j\in G}\in\mathbb C^{n_G\times2},
\end{equation}
whose rows $\beta_j$ are the rows of $\BF$ in the dual basis of $\cA'$.

\begin{corollary}[Cluster input matrices are invertible with norms growing linearly]\label{cor:BG}
Given the low block of Definition~\ref{def:P}, there are $C_L$ and $j_3\ge j_2$ such that for every $j>d$
\begin{equation}\label{eq:betaPhi}
\big|\beta_j-\big(b_1[X_j],b_2[X_j]\big)\big|\le\frac{C_L}{j},
\end{equation}
and for every high cluster $G$ of index $j\ge j_3$, $\sigma_{\min}(\mathbb B_G)\ge c_1j$ if $n_G=2$ and $|\mathbb B_G|\ge c_1j$ if $n_G=1$; consequently $\|\mathbb B_G^{(+)}\|\le1/(c_1j)$, with $\mathbb B_G^{(+)}$ the inverse for a pair and the minimum-norm right inverse for a singleton.
\end{corollary}

\subsection{Fredholm stage and proof of Theorem~\protect\ref{thm:main}}\label{sec:fredholm}

\begin{definition}[Target operator]\label{def:target}
$\widetilde\cS\coloneqq\operatorname{diag}(S_\alpha-q_1,\ \Lap-q_2)$ on $\cX$, with spectrum $\{\nu_n\coloneqq\lambda^S_n-q_1\}\cup\{\nu'_m\coloneqq\delta_m-q_2\}$ and the orthonormal uncoupled eigenvectors $\tilde e_\nu$. Each target pole $\nu$ is assigned to the cluster $G(\nu)$ of its uncoupled eigenvalue $\nu+q_i$; the low target poles are the shifted low uncoupled eigenvalues and are the $\mu_1,\dots,\mu_d$ of Definition~\ref{def:P}.
\end{definition}

\begin{proposition}[Low block, shifts and middle vectors of the design exist]\label{prop:design}
Let $\omega>0$, $\epsilon\in(0,\epsilon_0)$, $b$ the first cluster boundary at or above $j_1$, and $d\coloneqq b-1$. Then $\operatorname{span}\{X_1,\dots,X_d\}$ is a union of eigenspaces of $\cA$, and the low block contains every cluster of index below $j_1$; the set $\mathcal E$ of $(q_1,q_2)\in[\omega,\infty)^2$ for which the target spectrum is not simple or contains an eigenvalue of $\cA$ lies on countably many lines $q_1=a$, $q_2=a$, $q_1-q_2=a$, and its complement is dense in $[\omega,\infty)^2$; for $(q_1,q_2)\notin\mathcal E$ the target poles have real parts less than $-\omega$ and a gain $K_L$ as in Definition~\ref{def:P} exists; and in every middle cluster vectors $\hat c_\nu$ as in step~7 of Table~\ref{tab:fred} exist.
\end{proposition}

\begin{definition}[Closed-loop eigenvectors]\label{def:psi}
Fix $n_*\ge j_3$; the high clusters of index below $n_*$ are the middle clusters, the others the tail. For a target pole $\nu$ in a tail cluster $G=G(\nu)$, let $v_\nu\coloneqq Q_G\tilde e_\nu$, $M_G\coloneqq\cA'|_{Q_G\cX}=\cA|_{Q_G\cX}$, and
\begin{subequations}\label{eq:tailvec}
\begin{gather}
c_\nu\coloneqq\mathbb B_G^{(+)}(\nu-M_G)v_\nu\in\mathbb C^2,\label{eq:cnu}\\
\psi_\nu\coloneqq v_\nu+\sum_{j\notin G}\frac{\ip{\BF c_\nu}{\tilde X'_j}_\cX}{\nu-\lambda'_j}\,X'_j ,\label{eq:psinu}
\end{gather}
\end{subequations}
with $\lambda'_j$ the eigenvalue of $\cA'$ belonging to $X'_j$ ($\lambda'_j=\mu_j$ for $j\le d$, $\lambda_j$ for $j>d$). In \eqref{eq:psinu} $c_\nu\in\mathbb C^2$ is read as $(0,c_\nu)\in\mathbb C^3$, so that $\BF c_\nu=B(0,c_\nu)$. For the low target poles, $\psi_{\mu_j}\coloneqq X'_j$ and $c_{\mu_j}\coloneqq0$, $j\le d$. In each middle cluster take the bijection and the vectors of step~7 of Table~\ref{tab:fred}, written $k(\nu)\coloneqq p(k)$ and $\hat c_\nu\coloneqq\hat c_k$ for $\nu=\nu_k$. For $\xi\in\mathbb C$ put, in the middle clusters,
\begin{subequations}\label{eq:midvec}
\begin{gather}
\nu(\xi)\coloneqq\lambda_{k(\nu)}+\xi\big(\nu-\lambda_{k(\nu)}\big),\label{eq:family}\\
c_\nu(\xi)\coloneqq\big(\nu(\xi)-\lambda_{k(\nu)}\big)\hat c_\nu,\label{eq:familyc}\\
\psi_\nu(\xi)\coloneqq\sum_{j\ge1}\frac{\nu(\xi)-\lambda_{k(\nu)}}{\nu(\xi)-\lambda'_j}\,\ip{B\hat c_\nu}{\tilde X'_j}_\cX\,X'_j ,\label{eq:psimid}
\end{gather}
\end{subequations}
the coefficient $(\nu(\xi)-\lambda_{k(\nu)})/(\nu(\xi)-\lambda'_j)$ read as $1$ when $\lambda'_j=\lambda_{k(\nu)}$; in the tail $\nu(\xi)\coloneqq\nu$, $c_\nu(\xi)\coloneqq c_\nu$, $\psi_\nu(\xi)\coloneqq\psi_\nu$. $\widetilde\cS_\xi$ is the operator with the eigenvectors $\tilde e_\nu$ and the eigenvalues $\nu(\xi)$.
\end{definition}

\begin{lemma}[Prescribed states are closed-loop eigenvectors and are quadratically close to the uncoupled basis]\label{lem:eig}
For a design of Tables~\ref{tab:static} and~\ref{tab:fred}, for every high target pole $\nu$ and every $\xi$ with $\nu(\xi)\notin\sigma(\cA')\setminus\{\lambda_{k(\nu)}\}$ for the middle poles, the series \eqref{eq:psinu} and \eqref{eq:psimid} converge in $\cX$,
\begin{equation}\label{eq:eigid}
\big(\nu(\xi)-\cA'_{-1}\big)\psi_\nu(\xi)=Bc_\nu(\xi) ,
\end{equation}
and, for the tail poles, with $q\coloneqq\max(q_1,q_2)$ and $C$ depending on the design,
\begin{subequations}\label{eq:closesys}
\begin{gather}
|c_\nu|\le\frac{C(q+\kappa_0)}{j},\qquad
\|\psi_\nu-\tilde e_\nu\|^2\le\frac{C}{j^2},\label{eq:cbound}\\
\sum_\nu\|\psi_\nu-\tilde e_\nu\|^2<\infty .\label{eq:qc}
\end{gather}
\end{subequations}
\end{lemma}

\begin{proposition}[Transformation is invertible outside a finite set of parameter values]\label{prop:inv}
Define $R_\xi$ on $\cX$ by $R_\xi\tilde e_\nu\coloneqq\psi_\nu(\xi)$ for every target pole $\nu$, the low ones included, $R_\xi\tilde e_{\mu_j}=X'_j$. Let $M(\xi)$ be the matrix of $R_\xi$ from the high uncoupled eigenvectors to the high coordinates in the basis $\{X'_k\}$, with the high indices split into middle ($m$, the clusters of index below $n_*$) and tail ($\tau$), and
\begin{equation}\label{eq:Sigma}
\Sigma(\xi)\coloneqq M_{mm}(\xi)-M_{m\tau}M_{\tau\tau}^{-1}M_{\tau m}(\xi).
\end{equation}
For a design of Tables~\ref{tab:static} and~\ref{tab:fred} and $n_*$ large enough: (i) $M_{\tau\tau}$ does not depend on $\xi$ and is invertible; (ii) at every $\xi$ where the middle columns have no pole, $R_\xi-I$ is Hilbert--Schmidt, and $R_\xi$ is invertible if and only if $\det\Sigma(\xi)\ne0$; (iii) $\det\Sigma$ is meromorphic, that is, analytic except at isolated poles, on a neighborhood of $[0,2]$ with $\det\Sigma(0)\ne0$, so $R_\xi$ is invertible for all $\xi\in[0,2]$ outside a finite set, containing the poles of the middle columns and the zeros and poles of $\det\Sigma$, which contains $1$ only if $\det\Sigma(1)=0$. For such $\xi$, $\{\psi_\nu(\xi)\}$ is a Riesz basis of $\cX$, the feedback
\begin{equation}\label{eq:Fdef}
FX\coloneqq\sum_\nu c_\nu(\xi)\,\ip{X}{\tilde\psi_\nu}_\cX,\qquad\{\tilde\psi_\nu\}\ \text{the dual basis},
\end{equation}
is bounded from $\cX$ to $\mathbb C^3$, $(\cA'_{-1}+BF)\psi_\nu(\xi)=\nu(\xi)\psi_\nu(\xi)$ for every $\nu$, and $T\coloneqq R_\xi^{-1}$ satisfies $T-I\in\mathrm{HS}$.
\end{proposition}

\begin{proof}
Structure: since $R_\xi\tilde e_{\mu_j}=X'_j$, the low columns of the matrix of $R_\xi$ in the coordinates of $\{X'_k\}$ are unit vectors in the low coordinates, so this matrix is block upper triangular with the identity in the low block and $M(\xi)$ in the high block, and $R_\xi$ is invertible if and only if $M(\xi)$ is. (i) The tail columns do not involve $\xi$. $M_{\tau\tau}=E+N$, where $E$ is block diagonal with the blocks $E_G\coloneqq(\ip{\tilde e_\nu}{X_k}_\cX)_{k,\nu\in G}$, the matrix of $Q_G:Q^0_G\cX\to Q_G\cX$ in orthonormal bases, $\|E_G^{-1}\|\le2$ by \eqref{eq:QQ}, and $N$ collects the off-cluster tail coordinates, $\|N\|_{\mathrm{HS}}^2\le\sum_{\nu\in\tau}C|c_\nu|^2\le C\sum_{n\ge n_*}n^{-2}$ by \eqref{eq:tailsum} and \eqref{eq:cbound}; for $n_*$ large $\|E^{-1}N\|<1$. (ii) $R_\xi-I$ is Hilbert--Schmidt by \eqref{eq:qc} and the finitely many middle and low columns, and by (i) $M(\xi)$ is invertible if and only if its Schur complement $\Sigma(\xi)$ is. (iii) By \eqref{eq:psimid} the middle columns are meromorphic with values in $\ell^2$, with poles only where $\nu(\xi)\in\sigma(\cA')\setminus\{\lambda_{k(\nu)}\}$, and $\det\Sigma$ is meromorphic. At $\xi=0$, $\psi_\nu(0)=\sum_{\lambda_k=\lambda_{k(\nu)}}\ip{B\hat c_\nu}{\tilde X'_k}_\cX\,X'_k$, so $M_{\tau m}(0)=0$ and $M_{mm}(0)$ is block diagonal with the invertible matrices of step~7 of Table~\ref{tab:fred}, up to the ordering of indices. A meromorphic function that is not identically zero has finitely many zeros and poles in the compact set $[0,2]$, and $\xi=1$ is not a pole because the target poles are not eigenvalues of $\cA'$. For $\xi$ outside the finite set, $R_\xi$ is invertible, the image of the orthonormal basis is a Riesz basis, the dual basis is bounded, $\sum|c_\nu(\xi)|^2<\infty$ by \eqref{eq:cbound}, so $F$ is bounded with $F\psi_\nu(\xi)=c_\nu(\xi)$, and $(\cA'_{-1}+BF)\psi_\nu(\xi)=\nu(\xi)\psi_\nu(\xi)$ by \eqref{eq:eigid}.
\end{proof}

\begin{lemma}[Closed loop is similar to the target on its natural domain]\label{lem:domain}
Let $\xi$ be as in Proposition~\ref{prop:inv}, $R\coloneqq R_\xi$, $F$ as in \eqref{eq:Fdef}, $K\coloneqq\kappa+F$ and
\begin{equation}\label{eq:AK}
\begin{gathered}
A_Kx\coloneqq\cA_{-1}x+BKx,\\
\mathcal D(A_K)\coloneqq\{x\in\cX:\cA_{-1}x+BKx\in\cX\}.
\end{gathered}
\end{equation}
Then $A_K=R\widetilde\cS_\xi R^{-1}$ with $\mathcal D(A_K)=R\,\mathcal D(\widetilde\cS_\xi)$.
\end{lemma}

\begin{theorem}[Three wall inputs assign any decay rate at each wavenumber]\label{thm:main}
Fix $\Pra>0$, $\Ra>0$, a decay rate $\omega>0$, $\epsilon\in(0,\epsilon_0)$, and shifts $q_1,q_2\ge\omega$ as in step~3 of Table~\ref{tab:static}.
There are a low-block size $d$ and an index $n_0$ such that, for every middle-cluster bound $n_*\ge n_0$ and the gains of Tables~\ref{tab:static} and~\ref{tab:fred} with these parameters, vectors $\hat c_\nu$ as required in step~7 exist, and for every sufficiently small $\delta>0$ there is $\xi_*\in[1,1+\delta)$, with $\xi_*=1$ unless $R_1$ is not invertible, such that the gain functions \eqref{eq:gains} at $\xi=\xi_*$ are in $L^2(0,1)$, the closed loop \eqref{eq:plant}, \eqref{eq:controller} generates an analytic semigroup, and for every initial state $X(0)=(V,\theta)(0)\in\hat\cX$ with $V(0)(1)=K_V[P_0X(0)]$
\begin{multline}\label{eq:decay}
\|u(t)\|^2+\|V(t)\|^2+\|\theta(t)\|^2\\
\le C\,e^{-2\omega t}\big(\|u(0)\|^2+\|V(0)\|^2+\|\theta(0)\|^2\big),\qquad t\ge0.
\end{multline}
The transformation $T\coloneqq R_{\xi_*}^{-1}$ has the form \eqref{eq:transf}, is bounded and boundedly invertible on $\cX$, $T-I$ is Hilbert--Schmidt, $TP_0$ maps the physical state onto the target with the eigenvectors $\tilde e_k$ and the eigenvalues $\nu_k(\xi_*)$, which is \eqref{eq:target} for $\xi_*=1$, and the closed-loop spectrum is $\{\nu_k(\xi_*)\}_{k\ge1}$: the values $\lambda^S_n-q_1$ and $\delta_m-q_2$, of which finitely many are moved by $O(\delta)$, all simple with real parts at most $-\omega$. Here $\epsilon_0\coloneqq\tfrac12\pi\min(1,\sqrt\Pra)$; $R_\xi$ is the operator on $\cX$ with $R_\xi\tilde e_k\coloneqq\sum_{j\ge1}\Psi_{jk}(\xi)X'_j$ for every $k$, the columns $\Psi_{\cdot k}(\xi)$ of step~7 taken for all $j\ge1$; $\hat\cX\coloneqq\{(V,\theta)\in H^1(0,1)\times L^2(0,1):V(0)=0\}$ with the norm of Definition~\ref{def:ops}; $P_0$ is the orthogonal projection of $\hat\cX$ onto $\cX$; $C\coloneqq\chi\big(1+\alpha\coth\alpha\,\|K_V\|^2\big)\|R_{\xi_*}\|^2\|R_{\xi_*}^{-1}\|^2$; $\chi\coloneqq\max(1,\Pra\Ra)/\min(1,\Pra\Ra)$; and $K_V$ is the $V_c$ row of \eqref{eq:fb}.
\end{theorem}

\begin{proof}[Proof of Theorem~\ref{thm:main}]
Take $d$ as in Proposition~\ref{prop:design}, which also gives the vectors $\hat c_\nu$, and $n_0\ge j_3$ as large as Proposition~\ref{prop:inv} requires. In the basis $\{X'_j\}$ of \eqref{eq:eigApsys}, whose plant coordinates are the columns of $W$ in \eqref{eq:rc-W}, the columns $\Psi_{\cdot k}(\xi)$ of step~7 are the coordinates of $\psi_\nu(\xi)$ of Definition~\ref{def:psi}, $\nu=\nu_k$, since $\ip{Bc}{\tilde X'_j}_\cX=\beta^{(3)}_j\cdot c$ and $\ip{\tilde e_\nu}{X_j}_\cX=E_{jk}$; so $R_\xi$ of Theorem~\ref{thm:main} is the operator of Proposition~\ref{prop:inv}, and $\Xi(\xi)\Psi(\xi)^{-1}W^{-1}$ in \eqref{eq:rc-F} is the matrix in plant coordinates of $F$ in \eqref{eq:Fdef}. Take the finite exceptional set of Proposition~\ref{prop:inv} and $\xi_*\in[1,1+\delta)$ outside it, $\xi_*=1$ when $R_1$ is invertible, that is, when $\det\Sigma(1)\ne0$. The finitely many middle poles move by $|\xi_*-1|\,|\nu-\lambda_{k(\nu)}|$, less than the distance of the target spectrum from $-\omega$ and than the gaps of the target spectrum for $\delta$ small, so $\widetilde\cS_{\xi_*}$ has simple spectrum with real parts at most $-\omega$. With $R\coloneqq R_{\xi_*}$ and $F$ from \eqref{eq:Fdef}, Lemma~\ref{lem:domain} gives $A_K=R\widetilde\cS_{\xi_*}R^{-1}$, which generates the analytic semigroup $Re^{\widetilde\cS_{\xi_*}t}R^{-1}$ with norm at most $\|R\|\|R^{-1}\|e^{-\omega t}$. By Proposition~\ref{prop:plant}(iii), whose modal equations \eqref{eq:modal-ode} are the coordinates of $\dot{\tilde X}=\cA_{-1}\tilde X+B(V_c,U_c,T_c)^\top$ in the eigenbasis, $A_K$ governs the component $\tilde X=P_0X$ of the physical state. Conversely, for $x\in\mathcal D(A_K)$ and $w\coloneqq Kx$, let $L_w$ be a smooth state with the boundary data \eqref{eq:liftedbc}--\eqref{eq:liftedbc2} for the inputs $w$, as in the proof of Proposition~\ref{prop:Bm}(i); then, with $L_w=(\ell,\eta)$,
\begin{equation}\label{eq:Lw}
\cA_{-1}L_w+Bw=\begin{pmatrix}\Lap^{-1}\big(\Pra\Lap^2\ell-\Pra\Ra\,\alpha^2\eta\big)\\ \Lap\eta+\ell+w_Vh_\alpha\end{pmatrix}\in\cX ,
\end{equation}
$w_V$ the $V_c$ component of $w$, so $x-L_w\in\mathcal D(\cA)$ and $x$ has the wall data $w$ and solves \eqref{eq:liftsys}. The physical state is $X=L_K\tilde X$, $L_K x\coloneqq x+K_V[x](h_\alpha,0)$, a bounded isomorphism of $\cX$ onto $\mathcal G_K\coloneqq\{X\in\hat\cX:V(1)=K_V[P_0X]\}$ with inverse $P_0$, and the closed loop is the analytic semigroup $L_K e^{A_Kt}P_0$ on $\mathcal G_K$; so $\|X(t)\|^2\le(1+\alpha\coth\alpha\,\|K_V\|^2)\|\tilde X(t)\|^2$ and $\|\tilde X(0)\|\le\|X(0)\|$. With the energy identity of Proposition~\ref{prop:plant}(iii) and $\chi\coloneqq\max(1,\Pra\Ra)/\min(1,\Pra\Ra)$, \eqref{eq:decay} holds with $C=\chi(1+\alpha\coth\alpha\,\|K_V\|^2)\|R\|^2\|R^{-1}\|^2$. The rows of $K$ are the bounded functionals $X\mapsto\ip{X}{\phi_r}_\cX$ with
\begin{equation}\label{eq:phir}
\phi_r=\sum_{i\le d}\overline{(K_L)_{r+1,i}}\,X_i+\sum_\nu\overline{(c_\nu(\xi_*))_{r+1}}\,\tilde\psi_\nu ,
\end{equation}
$\{\tilde\psi_\nu\}$ the basis dual to $\{\psi_\nu(\xi_*)\}$, the series converging in $\cX$ by Proposition~\ref{prop:inv}; this is $\phi_r$ of step~9, so the gain functions \eqref{eq:gains} are in $L^2(0,1)$ and \eqref{eq:gains} is $\ip{X}{\phi_r}_\cX$ written with $V_y=-i\alpha u$; on the physical state it takes the same value because $\phi_r\in\cX$ is orthogonal to $(h_\alpha,0)$. Finally, $J(V,\theta)\coloneqq(V_y,\alpha V,\alpha\sqrt{\Pra\Ra}\,\theta)$ is an isometry of $\cX$ into $L^2(0,1)^3$, $J(T-I)J^*$ is Hilbert--Schmidt on $L^2(0,1)^3$ and hence an integral operator with a kernel in $L^2((0,1)^2)$, and its second and third components give \eqref{eq:transf}, applied to $P_0X$.
\end{proof}

\section{Closed loop in physical variables}\label{sec:phys}

\begin{proposition}[Inputs are real, and the mean fields and the unactuated wavenumbers decay at the prescribed rate]\label{prop:mean}
Fix $\Ra$, $\Pra$ and $L$.
\begin{itemize}
\item[(i)] The plant at $-\alpha$ is the complex conjugate of the plant at $\alpha$. If the gain functions at $-\alpha$ are the conjugates of those at $\alpha$, the physical inputs $U_c(x,t)$, $T_c(x,t)$, $V_c(x,t)$ are real for real fields, and the closed-loop spectrum at $-\alpha$ is that at $\alpha$.
\item[(ii)] At $\alpha=0$, $V\equiv0$ and the mean fields obey
\begin{subequations}\label{eq:meansys}
\begin{gather}
\bar u_t=\Pra\,\bar u_{yy},\qquad\bar u(0)=0,\quad\bar u(1)=U_{c,0},\label{eq:mean}\\
\bar\theta_t=\bar\theta_{yy},\qquad\bar\theta(0)=0,\quad\bar\theta(1)=T_{c,0},\label{eq:mean2}
\end{gather}
\end{subequations}
decoupled from each other and from every $\alpha\ne0$. With the feedback \eqref{eq:meanfbsys},

the Volterra transformations $\bar\Theta=\bar\theta-\int_0^yk\,\bar\theta$ and $\bar w=\bar u-\int_0^yk_u\,\bar u$ map \eqref{eq:meansys} onto $\bar\Theta_t=\bar\Theta_{yy}-q_0\bar\Theta$ and $\bar w_t=\Pra\bar w_{yy}-q_u\bar w$, which decay at rates $\pi^2+q_0$ and $\Pra\pi^2+q_u$; $\bar\theta$ and $\bar u$ decay at the same rates with constants the condition numbers of the two transformations.
\item[(iii)] $\lambda_1(\alpha)\le-\min(\Pra,1)(\pi^2+\alpha^2)+\kappa_0(\alpha)$, and $\kappa_0(\alpha)$ is bounded as $\alpha\to\infty$; so for every $\omega>0$ there is $M(\omega)$ with $\lambda_1(\alpha_m)\le-\omega$ for all $|m|>M(\omega)$.
\end{itemize}
\end{proposition}

\begin{proof}[Proof of Theorem~\ref{thm:phys}]
The Fourier decomposition in $x$ is orthogonal in $L^2$ and diagonalizes the linearized dynamics wavenumber by wavenumber. Each actuated wavenumber decays with its constant by Theorem~\ref{thm:main}; each unactuated one is a self-adjoint contraction in the weighted norm of $\cX$ with rate $\ge\omega$ by Proposition~\ref{prop:mean}(iii), hence decays in the unweighted norm with constant $\chi$; the mean fields by Proposition~\ref{prop:mean}(ii). The unactuated wavenumbers are self-adjoint with a uniform sector, so together they generate an analytic semigroup; adding the finitely many remaining wavenumbers, each analytic, preserves analyticity. Summing the squared norms gives \eqref{eq:physdecay}.
\end{proof}

\section{Conclusions}

The RB layer is the channel of \cite{NSF26} with the base flow removed and a heat equation attached. With the two-fluid channel of \cite{Yih26}, three plants have now been stabilized by the same Fredholm design, and the difficulty lies in a different place in each. In the channel the plant operator is not self-adjoint, and the Riesz basis of its root vectors is obtained by perturbation from the Stokes operator. In the two-fluid channel the unstable state is the interface, reached across a fluid; the eigenvalue sequences of the two layers, which come arbitrarily close when the layers are uncoupled, are uniformly separated once the layers are joined, so one Fredholm input suffices, as in the channel. In the RB layer the plant operator is self-adjoint and the basis property is immediate, but the velocity and temperature families are not separated by the coupling: their eigenvalues occur in close pairs, and the high spectrum is assigned in clusters of up to two, with two Fredholm inputs, $U_c$ and $T_c$, and a $2\times2$ transformation.

The number of Fredholm inputs is therefore determined by the high spectrum and not by the number of coupled subsystems. The two-fluid channel has three subsystems, two fluids and an interface, and needs one Fredholm input; the RB layer has two and needs two. In all three plants the normal wall velocity carries a finite-rank feedback acting on finitely many modes.

The target is determined by the coupling responsible for the instability. In the two-fluid channel the target keeps the transmission of stress between the fluids and deletes the coupling of the interface into them. In the RB layer neither the velocity nor the temperature is unstable alone, and the target deletes the coupling in both directions, leaving the Stokes and heat equations, each with its own decay rate.

For the thermal convection loop, \cite{VK06} proved stability when the viscosity is large relative to the thermal diffusivity, that is, for large Prandtl number, by a singular perturbation that leaves a single equation for the temperature, stabilized by a Volterra design. Here no separation of time scales is needed: the velocity and the temperature are treated together by one Fredholm transformation, for every ratio of viscosity to thermal diffusivity.

The design uses the following structure of the RB layer: at each wavenumber, two scalar parabolic equations on an interval, of fourth and second order in space, each stable alone, coupled through bounded terms that close a destabilizing loop, with boundary inputs acting on both and with eigenvalue families of different densities that come close in pairs. The same structure, sometimes with more than two subsystems, appears in problems far from the heated layer. In double-diffusive convection \cite{Stern60}, which produces salt fingers in the ocean, the Boussinesq equations \eqref{eq:NS} carry a second scalar, the salinity, with a diffusivity of its own. In convection through a porous medium, the Darcy--Boussinesq equations \cite{HR45,Lapwood48} replace the fourth-order equation of the velocity by Darcy's law. In the Cahn--Hilliard--Navier--Stokes model of two fluids that separate, whose boundary stabilization in a channel is in \cite{Mun21}, the fourth-order Cahn--Hilliard equation takes the place of the heat equation. Outside fluid mechanics, the linearized Keller--Segel model of chemotaxis \cite{KS70} couples the diffusion of a cell density with that of a chemical attracting the cells; each is stable alone, and the loop between them produces aggregation. These plants differ from the RB layer in the orders of their equations, the number of their eigenvalue families, or the symmetry of their coupling, and each requires its own analysis. However, for the features they share with the RB layer, the tools developed here are the clustering of eigenvalues from different families and the Fredholm transformation with one boundary input per family.

A plant from a physical domain far from thermal convection has the structure of the RB layer together with the shear flow of the channel of \cite{NSF26}. Boundary control of the Hartmann flow, a conducting fluid driven between plates across an imposed transverse magnetic field \cite{Hartmann37}, was designed by Volterra backstepping in \cite{XSVK08} in the inductionless approximation, valid at small magnetic Reynolds number, where the magnetic field is not a state. At magnetic Reynolds number of order one the magnetic field is a state; in two dimensions it is represented by a flux function $A$, with $(B_x,B_y)=(A_y,-A_x)$, and resistive magnetohydrodynamics takes the form
\begin{subequations}\label{eq:mhd}
\begin{align}
\omega_t+\psi_y\omega_x-\psi_x\omega_y&=\frac1{\mathrm{Re}}\Delta\omega\nonumber\\&\quad+N\big(A_y(\Delta A)_x-A_x(\Delta A)_y\big),\label{eq:mhdvort}\\
A_t+\psi_yA_x-\psi_xA_y&=\frac1{\mathrm{Rm}}\Delta A,\label{eq:mhdflux}
\end{align}
\end{subequations}
with $\omega$, $\psi$ as in \eqref{eq:nl}, the Reynolds and magnetic Reynolds numbers $\mathrm{Re}$, $\mathrm{Rm}$, and the interaction parameter $N$. The flux function plays the role of the temperature: with the imposed field $B_y=1$, $A=-x+a$, the perturbation $a$ obeys \eqref{eq:temp} with $\psi_y$ in place of $-\psi_x$ and the diffusivity $1/\mathrm{Rm}$, and the Lorentz term in \eqref{eq:mhdvort} contains $N(\Delta a)_y$ in place of $-\Pra\Ra\,\theta_x$ in \eqref{eq:vort}. The inputs are the two wall velocities and the wall flux, set by electromagnets outside the plate. Unlike the RB layer, the Hartmann flow carries a net flow, so the plant has the shear terms of the channel of \cite{NSF26} in addition to the two-way coupling of the RB layer; its stabilization requires merging the construction of \cite{NSF26} with that of this paper, which is not carried out here. The model extends the inductionless one of \cite{XSVK08} to one in which the magnetic field is a state. With an equilibrium field that reverses direction across the layer, the equations \eqref{eq:mhd} also govern the tearing mode \cite{FKR63}, the resistive instability of tokamak plasmas in which magnetic field lines break and reconnect at the surface where the equilibrium field is orthogonal to the wavenumber. The tearing mode exists only because of the magnetic diffusivity $1/\mathrm{Rm}$, and is absent from the inductionless model.

\appendix

\section{Proofs of supporting results}\label{app:proofs}

\subsection{Plant, uncoupled spectra, and clustering}

\begin{proof}[Proof of Proposition~\ref{prop:plant}(i)]
In the Hilbert space $H^1_0$ with the $\alpha$-inner product, $-S_\alpha$ is the operator of the closed positive form
\begin{equation}\label{eq:aform}
a(V,W)\coloneqq\Pra\int_0^1\Lap V\,\overline{\Lap W}\,\dd y,\qquad V,W\in H^2_0(0,1),
\end{equation}
with $\ip{S_\alpha V}{W}_\alpha=-a(V,W)$; so $S_\alpha$ is self-adjoint and negative, with domain $H^3\cap H^2_0$ by one-dimensional elliptic regularity, and compact resolvent by the compact embedding $H^2\hookrightarrow H^1$. The Dirichlet Laplacian is self-adjoint and negative on $L^2$ with compact resolvent. $\cA=\cS+\mathcal K$ with $\cS\coloneqq\operatorname{diag}(S_\alpha,\Lap)$ and $\mathcal K$ bounded and symmetric on $\cX$ (Proposition~\ref{prop:cluster}(i), whose proof uses only the norm of Definition~\ref{def:ops}), so $\cA$ is self-adjoint with compact resolvent. Reflection invariance follows from the constancy of every coefficient in \eqref{eq:plant}.
\end{proof}

\begin{proof}[Proof of Proposition~\ref{prop:plant}(ii)]
Multiply the four equations
\begin{subequations}\label{eq:four}
\begin{gather}
s\Lap V-\Pra\Lap^2V+\Pra\Ra\,\alpha^2\theta=0,\label{eq:four1}\\
\lambda\Lap e-\Pra\Lap^2e+\Pra\Ra\,\alpha^2\vartheta=0,\label{eq:four1e}\\
s\theta-\Lap\theta-V=0,\qquad \lambda\vartheta-\Lap\vartheta-e=0,\label{eq:four2}
\end{gather}
\end{subequations}
by $-e$, $-V$, $\Pra\Ra\,\alpha^2\vartheta$, $\Pra\Ra\,\alpha^2\theta$ respectively; integrate and subtract eigenvector from resolvent in each pair. The four cross terms cancel. With $\int(\Lap f)g-\int f(\Lap g)=[f'g-fg']_0^1$ applied once to the second-order and twice to the fourth-order terms,
\begin{multline}\label{eq:lagr}
(s-\lambda)\ip{(V,\theta)}{(e,\vartheta)}_\cX
=-\lambda[e'V-eV']_0^1\\
-\Pra\big[(\Lap V)'e-(\Lap V)e'-(\Lap e)'V+(\Lap e)V'\big]_0^1\\
+\Pra\Ra\,\alpha^2[\theta'\vartheta-\theta\vartheta']_0^1 .
\end{multline}
Every term vanishes at $y=0$. At $y=1$, $e(1)=e'(1)=\vartheta(1)=0$ leave
\begin{equation}\label{eq:wallterms}
\begin{split}
&-\Pra\big[-(\Lap e)'(1)V_c+(\Lap e)(1)V_y(1)\big]\\
&\quad-\Pra\Ra\,\alpha^2T_c\,\vartheta'(1),
\end{split}
\end{equation}
and since $e(1)=e'(1)=0$, $(\Lap e)(1)=e''(1)$, $(\Lap e)'(1)=e'''(1)$; with $V_y(1)=-i\alpha U_c$ this is \eqref{eq:residue}.
\end{proof}

\begin{proof}[Proof of Proposition~\ref{prop:plant}(iii)]
For $V\in H^1_0$, $\ip{V}{h_\alpha}_\alpha=[Vh_\alpha']_0^1-\int_0^1V\Lap h_\alpha\,\dd y=0$, since $V$ vanishes at the walls and $\Lap h_\alpha=0$; $\|h_\alpha\|_\alpha^2=[h_\alpha h_\alpha']_0^1=\alpha\coth\alpha$; and $V=(V-V(1)h_\alpha)+V(1)h_\alpha$. The energy identity is $u=V_y/(-i\alpha)$. With $\tilde V\coloneqq V-V_ch_\alpha$, $\Lap h_\alpha=0$ removes $\dot V_c$ and $V_c$ from the velocity equation:
\begin{subequations}\label{eq:liftsys}
\begin{gather}
\Lap\tilde V_t=\Pra\Lap^2\tilde V-\Pra\Ra\,\alpha^2\theta,\label{eq:lifted}\\
\theta_t=\Lap\theta+\tilde V+V_ch_\alpha,\label{eq:liftedT}\\
\tilde V(0)=\tilde V(1)=0,\qquad\tilde V'(0)=-\frac{\alpha}{\sinh\alpha}V_c,\label{eq:liftedbc}\\
\tilde V'(1)=-i\alpha U_c-\alpha\coth\alpha\,V_c,\label{eq:liftedbc1}\\
\theta(0)=0,\qquad\theta(1)=T_c .\label{eq:liftedbc2}
\end{gather}
\end{subequations}
The computation of Proposition~\ref{prop:plant}(ii), with the time derivative in place of $s$ and $X_j=(e,\vartheta)$ real, gives
\begin{multline}\label{eq:liftmodal}
\frac{\dd}{\dd t}\ip{\tilde X}{X_j}_\cX-\lambda_j\ip{\tilde X}{X_j}_\cX\\
=i\alpha\Pra\,e''(1)U_c-\Pra\Ra\,\alpha^2\vartheta'(1)T_c
+\Pra\Big(\alpha\coth\alpha\,e''(1)\\
-\frac{\alpha}{\sinh\alpha}\,e''(0)+\Ra\,\alpha^2\int_0^1\vartheta\,h_\alpha\,\dd y\Big)V_c ,
\end{multline}
the last term collecting the two wall slopes of $\tilde V$ and the source $V_ch_\alpha$. Integrating $\Pra\Lap^2e-\lambda_j\Lap e=\Pra\Ra\,\alpha^2\vartheta$ against $h_\alpha$, with $\int_0^1(\Lap e)h_\alpha\,\dd y=0$ by the clamped conditions,
\begin{equation}\label{eq:b0id}
e'''(1)=\alpha\coth\alpha\ e''(1)-\frac{\alpha}{\sinh\alpha}\,e''(0)+\Ra\,\alpha^2\int_0^1\vartheta\,h_\alpha\,\dd y ,
\end{equation}
so the coefficient of $V_c$ in \eqref{eq:liftmodal} is $b_0[X_j]$. Finally $\ip{\tilde X}{X_j}_\cX=\ip{X}{X_j}_\cX$ by the orthogonality.
\end{proof}

\begin{proof}[Proof of Proposition~\ref{prop:joint}(i)]
Solutions of $\Lap(\Pra\Lap-\lambda)V=0$ with $\lambda=-\Pra(\sigma^2+\alpha^2)$ are combinations of $\cosh\alpha z,\sinh\alpha z,\cos\sigma z,\sin\sigma z$, $z=y-\tfrac12$. On the symmetric branch $V=A\cosh\alpha z+C\cos\sigma z$ and $V(\tfrac12)=V'(\tfrac12)=0$ give $A\cosh\tfrac\alpha2+C\cos\tfrac\sigma2=0$, $A\alpha\sinh\tfrac\alpha2-C\sigma\sin\tfrac\sigma2=0$, whence \eqref{eq:roots}; the antisymmetric branch gives \eqref{eq:roots2}. With $\sigma=2s$, \eqref{eq:roots} reads $\tan s=-(\alpha/\sigma)\tanh\tfrac\alpha2$ and \eqref{eq:roots2} reads $\tan s=(\sigma/\alpha)\tanh\tfrac\alpha2$, giving \eqref{eq:sigma}. For \eqref{eq:etr}: with $A=-\cos\tfrac\sigma2/\cosh\tfrac\alpha2$ and the $\Ha$-normalized amplitude $\sqrt2/\sigma$,
\begin{subequations}\label{eq:e23}
\begin{align}
e_n''(1)&=\sqrt2\sigma^{-1}\big[A\alpha^2\cosh\tfrac\alpha2-\sigma^2\cos\tfrac\sigma2\big]\simeq-\sqrt2\,\sigma\cos\tfrac\sigma2,\label{eq:e2}\\
e_n'''(1)&=\sqrt2\sigma^{-1}\big[A\alpha^3\sinh\tfrac\alpha2+\sigma^3\sin\tfrac\sigma2\big]\nonumber\\
&=-\sqrt2\sigma^{-1}\alpha\tanh\tfrac\alpha2\cos\tfrac\sigma2\,(\alpha^2+\sigma^2),\label{eq:e3}
\end{align}
\end{subequations}
the last equality by \eqref{eq:roots}; two powers of $\sigma$ cancel. The antisymmetric branch is the same with $\coth$. \eqref{eq:ttr} is direct.
\end{proof}

\begin{proof}[Proof of Proposition~\ref{prop:joint}(ii)--(iii)]
\eqref{eq:jidx}: $|\lambda^S_n|=\Pra\pi^2n^2(1+O(1/n))$ and $|\delta_m|=\pi^2m^2(1+O(1/m))$, so the number of $\delta_m$ above $\lambda^S_n$ is $\Pra^{1/2}n+O(1)$ and the number of $\lambda^S_n$ above $\delta_m$ is $\Pra^{-1/2}m+O(1)$. \eqref{eq:bsys} are \eqref{eq:functionals} with \eqref{eq:trsys}. For the flux case the trace identity replaces $-\Pra\Ra\,\alpha^2\vartheta'(1)T_c$ by $-\Pra\Ra\,\alpha^2\vartheta(1)q_c/k_{\rm th}$, and $|\vartheta(1)|=\sqrt2/(\alpha\sqrt{\Pra\Ra})$ on the normalized uncoupled eigenvectors.
\end{proof}

\begin{proof}[Proof of Proposition~\ref{prop:cluster}(i)]
Two estimates:
\begin{subequations}\label{eq:Ksys}
\begin{gather}
\begin{aligned}\|(\Pra\Ra\,\alpha^2\Lap^{-1}\theta,0)\|^2_\cX&=(\Pra\Ra\,\alpha^2)^2\ip{\theta}{-\Lap^{-1}\theta}\\&\le\kappa_0^2\,\|(0,\theta)\|^2_\cX,\end{aligned}\label{eq:K1}\\
\begin{aligned}\|(0,V)\|_\cX^2=\Pra\Ra\,\alpha^2\|V\|^2&\le\frac{\Pra\Ra\,\alpha^2\|V\|_\alpha^2}{\pi^2+\alpha^2}\\&=\kappa_0^2\|(V,0)\|^2_\cX,\end{aligned}\label{eq:K2}
\end{gather}
\end{subequations}
using $\|V\|_\alpha^2\ge(\pi^2+\alpha^2)\|V\|^2$ for $V\in H^1_0$ and $\ip{\theta}{-\Lap^{-1}\theta}\le\|\theta\|^2/(\pi^2+\alpha^2)$. Symmetry: with $w\coloneqq\Pra\Ra\,\alpha^2$, $\ip{\mathcal K(V,\theta)}{(Z,\eta)}_\cX=w\ip{\theta}{Z}+w\ip{V}{\eta}$, symmetric in the two arguments. The eigenvalue bound is the min-max principle.
\end{proof}

\begin{proof}[Proof of Proposition~\ref{prop:cluster}(ii)]
Consecutive eigenvalues of one family are separated by $\Pra(\sigma_{n+1}^2-\sigma_n^2)=2\pi\sqrt{\Pra|\lambda^S_n|}\,(1+o(1))$ and $\delta_m-\delta_{m+1}=2\pi\sqrt{|\delta_m|}\,(1+o(1))$. A component of three elements contains two of the same family at distance at most the sum of two joined edges, $\le2\epsilon\sqrt{|\varsigma|}(1+o(1))<4\epsilon_0\sqrt{|\varsigma|}=2\pi\min(1,\sqrt{\Pra})\sqrt{|\varsigma|}$, below the same-family separation for $j$ large: contradiction. So components have at most two elements, of different families, and the diameter bound is the joining rule. The nearest eigenvalue outside a component is not joined to it, so its distance from the nearer endpoint exceeds $\epsilon\min(\sqrt{|\cdot|},\sqrt{|\cdot|})=\epsilon\sqrt{|\varsigma_j|}(1-o(1))$. The last claim is \eqref{eq:jidxsys}.
\end{proof}

\begin{proof}[Proof of Proposition~\ref{prop:cluster}(iii)]
By Proposition~\ref{prop:cluster}(i) every eigenvalue of $\cA$ is within $\kappa_0$ of one of $\cS$, so for $\tfrac12g_G\ge2\kappa_0$ the circle is at distance $\ge\tfrac14g_G$ from $\sigma(\cA)$ and $\sigma(\cS)$, and the resolvents are bounded by the reciprocal of that distance on $\Gamma_G$ by self-adjointness. Then
\begin{subequations}\label{eq:QQsys}
\begin{gather}
Q_G-Q^0_G=\frac1{2\pi i}\oint_{\Gamma_G}(z-\cA)^{-1}\mathcal K(z-\cS)^{-1}\dd z,\label{eq:QQint}\\
\|Q_G-Q^0_G\|\le r_G\cdot\kappa_0\cdot\Big(\frac{4}{g_G}\Big)^2\le\frac{C_G\kappa_0}{j},\label{eq:QQbound}
\end{gather}
\end{subequations}
by \eqref{eq:rG}; the bound is below $\tfrac12$ for $j\ge j_1$, so the ranks agree. Summability follows, and the last claim is $\|(I-Q_G)\tilde e_\nu\|=\|(Q^0_G-Q_G)\tilde e_\nu\|$.
\end{proof}

\subsection{Wall responses, Fredholm stage, and physical variables}

\begin{proof}[Proof of Proposition~\ref{prop:Bm}(i)]
Existence and uniqueness: with the smooth state $L_1\coloneqq\big(-i\alpha(y^3-y^2),0\big)$ or $L_2\coloneqq(0,y)$, which has the wall datum and satisfies the conditions at $y=0$ and the homogeneous conditions at $y=1$, $Y_{i,s}=L_i+(s-\cA)^{-1}f_i$, where $f_i\in\cX$ is the residual of $L_i$ in the resolvent equations; the difference of two solutions lies in $\mathcal D(\cA)$ and in the kernel of $s-\cA$. Identity: the computation in the proof of Proposition~\ref{prop:plant}(ii) uses of $(e,\vartheta)$ only its boundary conditions and the value of $\cA(e,\vartheta)$; repeated with $g=(\zeta,\eta)\in\mathcal D(\cA)$ in place of $(e,\vartheta)$, it gives $\ip{Y_{1,s}}{(\bar s-\cA)g}=-\overline{b_1[g]}$ and $\ip{Y_{2,s}}{(\bar s-\cA)g}=\overline{b_2[g]}$, first for smooth $g$ and then for $g\in\mathcal D(\cA)$ by density in the graph norm, since the terms in $\zeta'''$ have the factor $V_{i,s}(0)=V_{i,s}(1)=0$ of the velocity component of $Y_{i,s}$ and $\zeta''(1)$, $\eta'(1)$ are continuous on $H^3\times H^2$; $s=\bar z$, $g=(z-\cA)^{-1}f$ give \eqref{eq:dual}. The same computation holds for $\cS$.

Uncoupled bound: with $g=\tilde e_k$ the identity for $\cS$ gives $|\ip{Y^0_{i,s}}{\tilde e_k}_\cX|=|b_i[\tilde e_k]|/|s-\varsigma_k|$; $b_1$ vanishes on the uncoupled thermal eigenvectors and $b_2$ on the velocity ones, so by Parseval
\begin{equation}\label{eq:Y0}
\begin{gathered}
\|Y^0_{1,s}\|^2=\sum_{n\ge1}\frac{|b_1[e_n]|^2}{|s-\lambda^S_n|^2},\\
\|Y^0_{2,s}\|^2=\sum_{m\ge1}\frac{|b_2[\theta_m]|^2}{|s-\delta_m|^2}.
\end{gathered}
\end{equation}
Take the first sum; let $n_*$ minimize $|\lambda^S_n-z_G|$, with $z_G$ the center of $\Gamma_G$, so that $n_*\asymp j$. For $s\in\Gamma_G$, $|s-\lambda^S_n|\ge\tfrac12g_G\ge cj$ for every $n$ by \eqref{eq:rG} and Proposition~\ref{prop:cluster}(ii). Since $|s-\lambda^S_{n_*}|\le|\lambda^S_{n_*}-z_G|+r_G\le Cj$ and $|\lambda^S_n-\lambda^S_{n_*}|\ge c|n^2-n_*^2|\ge c|n-n_*|\,n_*$, there is $K$ independent of $j$ with $|s-\lambda^S_n|\ge\tfrac12|\lambda^S_n-\lambda^S_{n_*}|$ for $|n-n_*|\ge K$. With $|b_1[e_n]|\le Cn$ from \eqref{eq:bgrow},
\begin{equation}\label{eq:Y0bound}
\|Y^0_{1,s}\|^2\le\sum_{|n-n_*|<K}\frac{Cn^2}{c^2j^2}+\sum_{|n-n_*|\ge K}\frac{4Cn^2}{c^2(n^2-n_*^2)^2},
\end{equation}
where the first sum has at most $2K$ terms, each with $n\asymp j$, and the second is bounded uniformly in $n_*$ as in \eqref{eq:tailsum}. The thermal sum is the same with $\delta_m$ and $|b_2[\theta_m]|=\sqrt2\,\pi\alpha\sqrt{\Pra\Ra}\,m$.

Comparison: $\cA$ and $\cS$ differ by the bounded operator $\mathcal K$, which involves no derivatives, so $Y_{i,s}-Y^0_{i,s}$ has zero wall data and $(s-\cA)(Y_{i,s}-Y^0_{i,s})=\mathcal KY^0_{i,s}$. By the proof of Proposition~\ref{prop:cluster}(iii), $\|(s-\cA)^{-1}\|\le4/g_G\le1/\kappa_0$ on $\Gamma_G$ for $j\ge j_1$, which gives the factor $2$ in \eqref{eq:respbound}.
\end{proof}

\begin{proof}[Proof of Corollary~\ref{cor:BG}]
For $j>d$ the dual vector is $\tilde X'_j=X_j+\sum_{i\le d}a_{ji}X_i$, and biorthogonality with the low eigenvectors \eqref{eq:eigAp2}, whose coordinate on $X_j$ is $w_{jr}\coloneqq(b_j\cdot c_r)/(\mu_r-\lambda_j)$, gives $X_L^\top\overline{a_j}=-w_{j\cdot}$. By Proposition~\ref{prop:Bm}(ii), $|w_{jr}|\le C_1j|c_r|/|\mu_r-\lambda_j|\le C/j$, since $|\mu_r-\lambda_j|\ge\tfrac12|\lambda_j|\ge cj^2$ when $|\lambda_j|\ge2\max_r|\mu_r|$, and $\mu_r\in\rho(\cA)$ for the finitely many other $j$. Then $|\beta_j-(b_1[X_j],b_2[X_j])|\le|a_j|\,\|B_L\|\le C_L/j$, and Proposition~\ref{prop:Bm}(iii) gives the lower bounds once $2c_1j-\sqrt2\,C_L/j\ge c_1j$.
\end{proof}

\begin{proof}[Proof of Proposition~\ref{prop:design}]
By Proposition~\ref{prop:cluster}(i), $|\lambda_j-\varsigma_j|\le\kappa_0$, and at the boundary $b\ge j_1$ the gap exceeds $4\kappa_0$, so $\lambda_d>\lambda_{d+1}$. For $q_i\ge\omega$ the target poles are at most $\varsigma_1-\omega<-\omega$. Simplicity of the target spectrum and distinctness of the low target poles exclude $q_1-q_2\in\{\lambda^S_n-\delta_m\}$; $\mu_j\in\rho(\cA)$ excludes, for each $j\le d$, countably many values of one shift; a high target pole $\varsigma_k-q_i$ equal to some $\lambda_j$ excludes one value of $q_i$ for each pair $(k,j)$. A countable union of lines is meagre in the plane. $K_L$ exists by Proposition~\ref{prop:low}. For $\hat c_\nu$: if $w$ satisfies $\cA'^*w=\bar\lambda w$ and $\ip{Bc}{w}_\cX=0$ for every $c\in\mathbb C^3$, then $\cA w=\cA'^*w-\kappa^*B^*w=\lambda w$ with all three wall functionals of $w$ zero, so $w=0$ by Theorem~\ref{thm:three}; hence the matrix $\big(\ip{Be_i}{\tilde X'_k}_\cX\big)$, $i=1,2,3$, $\lambda_k=\lambda$, has full rank, and $\hat c_\nu$ can be chosen to make each square matrix of step~7 invertible.
\end{proof}

\begin{proof}[Proof of Lemma~\ref{lem:eig}]
Middle poles: $\psi_\nu(\xi)$ is the expansion of $(\nu(\xi)-\lambda_{k(\nu)})(\nu(\xi)-\cA'_{-1})^{-1}B\hat c_\nu$, the pole at $\lambda_{k(\nu)}$ removed, and its coefficients are $O(1/j)$ by Proposition~\ref{prop:Bm}(ii) and \eqref{eq:betaPhi}. Tail poles: $Q_G\cX$ is $\cA'$-invariant and $\cA'=\cA$ there, so $(\nu-\cA')v_\nu=(\nu-M_G)v_\nu$, and by \eqref{eq:cnu} and \eqref{eq:BG} this equals the cluster coordinates of $\BF c_\nu$ in the eigenbasis of $\cA'$, namely $\sum_{j\in G}\ip{\BF c_\nu}{\tilde X'_j}_\cX\,X'_j$. (What is prescribed is the cluster coordinate of $\psi_\nu$ in the eigenbasis of $\cA'$, not $Q_G\psi_\nu$, which also contains the tails of the low eigenvectors.) Applying $(\nu-\cA'_{-1})$ term by term to the series in \eqref{eq:psinu} returns $\sum_{j\notin G}\ip{\BF c_\nu}{\tilde X'_j}_\cX\,X'_j$. The sum of the two is the expansion of $\BF c_\nu\in\cX_{-1}$ in $\{X'_j\}$, valid in $\cX_{-1}$ because the change of basis differs from the identity by finitely many columns in $\cX$ and coordinate functionals $\ip{\cdot}{X_i}_\cX$, $X_i\in\mathcal D(\cA)$; this is \eqref{eq:eigid}. Bound on $c_\nu$: $\|(\nu-M_G)v_\nu\|\le\|(\nu-\cA)Q_G\tilde e_\nu\|=\|Q_G(\nu-\varsigma-\mathcal K)\tilde e_\nu\|\le q+\kappa_0$, with $\varsigma$ the uncoupled eigenvalue of $\tilde e_\nu$, and Corollary~\ref{cor:BG}, since $n_*\ge j_3$. Closeness: $\psi_\nu-\tilde e_\nu=-(I-Q_G)\tilde e_\nu+\sum_{j\notin G}(\cdots)$; the first term is bounded by \eqref{eq:QQ}. In the second, for $j>d$ the coefficient $|\ip{\BF c_\nu}{\tilde X'_j}_\cX|=|\beta_j\cdot c_\nu|\le C|c_\nu|\,j$ by Proposition~\ref{prop:Bm}(ii) and \eqref{eq:betaPhi}, and $|\nu-\lambda_j|\ge c|j^2-n^2|$ for $j\notin G$, where $n$ is the joint index of the uncoupled eigenvalue $\nu+q_i$: for $j$ beyond a cutoff depending on $q$ and $\kappa_0$ this is Proposition~\ref{prop:cluster}(ii), and for the finitely many remaining $j$ it holds with a smaller $c$ because the target poles are not eigenvalues of $\cA'$; so
\begin{equation}\label{eq:tailsum}
\begin{split}
\Big\|\sum_{j\notin G,\,j>d}\frac{\ip{\BF c_\nu}{\tilde X'_j}_\cX}{\nu-\lambda_j}X'_j\Big\|^2
&\le C|c_\nu|^2\sum_{j\ne n}\frac{j^2}{(j^2-n^2)^2}\\
&\le C'|c_\nu|^2 ,
\end{split}
\end{equation}
since $\sum_{j\ne n}j^2/(j^2-n^2)^2\le\sum_{k\ne0}k^{-2}+\sum_{j>2n}j^{-2}+n^{-4}\sum_{j<n/2}j^2$ is bounded uniformly in $n$; the finitely many $j\le d$ contribute a bounded multiple of $|c_\nu|^2$ because the target spectrum is simple. With \eqref{eq:cbound} for $c_\nu$ this gives $\|\psi_\nu-\tilde e_\nu\|^2\le C/j^2$, summable.
\end{proof}

\begin{proof}[Proof of Lemma~\ref{lem:domain}]
For real $z>0$ let $G_z\coloneqq(z-\cA_{-1})^{-1}B:\mathbb C^3\to\cX$. In the plant eigenbasis, $\|G_z\|^2\le\sum_j|b_j|^2/(z-\lambda_j)^2\le C\sum_jj^2/(z+cj^2)^2=O(z^{-1/2})$ by Proposition~\ref{prop:Bm}(ii), splitting the sum at $j\asymp\sqrt z$. So for $z$ large $I-G_zK$ is invertible and $(z-A_K)^{-1}=(I-G_zK)^{-1}(z-\cA)^{-1}$ is a bounded resolvent with range in $\mathcal D(A_K)$. Next, $(\cA_{-1}+BK)R\tilde e_\nu=R\widetilde\cS_\xi\tilde e_\nu$ for every target eigenvector, by \eqref{eq:eigid} and $F\psi_\nu(\xi)=c_\nu(\xi)$; finite sums of target eigenvectors are a core for $\widetilde\cS_\xi$, and $\cA_{-1}+BK$ is bounded from $\cX$ to $\cX_{-1}$, so passing to the limit gives $(\cA_{-1}+BK)Ry=R\widetilde\cS_ty\in\cX$ for $y\in\mathcal D(\widetilde\cS_\xi)$, that is, $R\widetilde\cS_tR^{-1}\subset A_K$. Both operators have $z$ in their resolvent sets; for $x\in\mathcal D(A_K)$ put $f\coloneqq(z-A_K)x$ and $y\coloneqq(z-R\widetilde\cS_tR^{-1})^{-1}f$; then $(z-A_K)y=f=(z-A_K)x$ and injectivity of $z-A_K$ gives $x=y\in R\,\mathcal D(\widetilde\cS_\xi)$.
\end{proof}

\begin{proof}[Proof of Proposition~\ref{prop:mean}(ii)]
Averaging \eqref{eq:NSu} and \eqref{eq:NST} in $x$, with $v$ of zero mean by \eqref{eq:NSdiv} and $v(0)=0$, gives \eqref{eq:meansys} at linear order. The kernels are the classical ones for the heat equation with Dirichlet actuation and target shift, with $I_1(s)/s\to\tfrac12$ on the diagonal.
\end{proof}

\begin{proof}[Proof of Proposition~\ref{prop:mean}(iii)]
By Proposition~\ref{prop:cluster}(i) and the lowest eigenvalues of $S_\alpha$ and $\Lap$,
\begin{equation}\label{eq:rayq}
\begin{split}
\lambda_1&=\sup_X\frac{\ip{\cS X}{X}_\cX+\ip{\mathcal KX}{X}_\cX}{\|X\|_\cX^2}\\
&\le-\min(\Pra,1)(\pi^2+\alpha^2)+\kappa_0 .
\end{split}
\end{equation}
\end{proof}

\section*{Acknowledgment}
The authors' problems, ideas, and results were developed with the assistance of Claude and ChatGPT in final theorem formulation, proofs, simulations, and drafting throughout the paper, under the authors' correction and complete verification.

\vspace{-.5cm}
\begin{IEEEbiography}[{\includegraphics[width=1in,height=1.25in,clip,keepaspectratio]{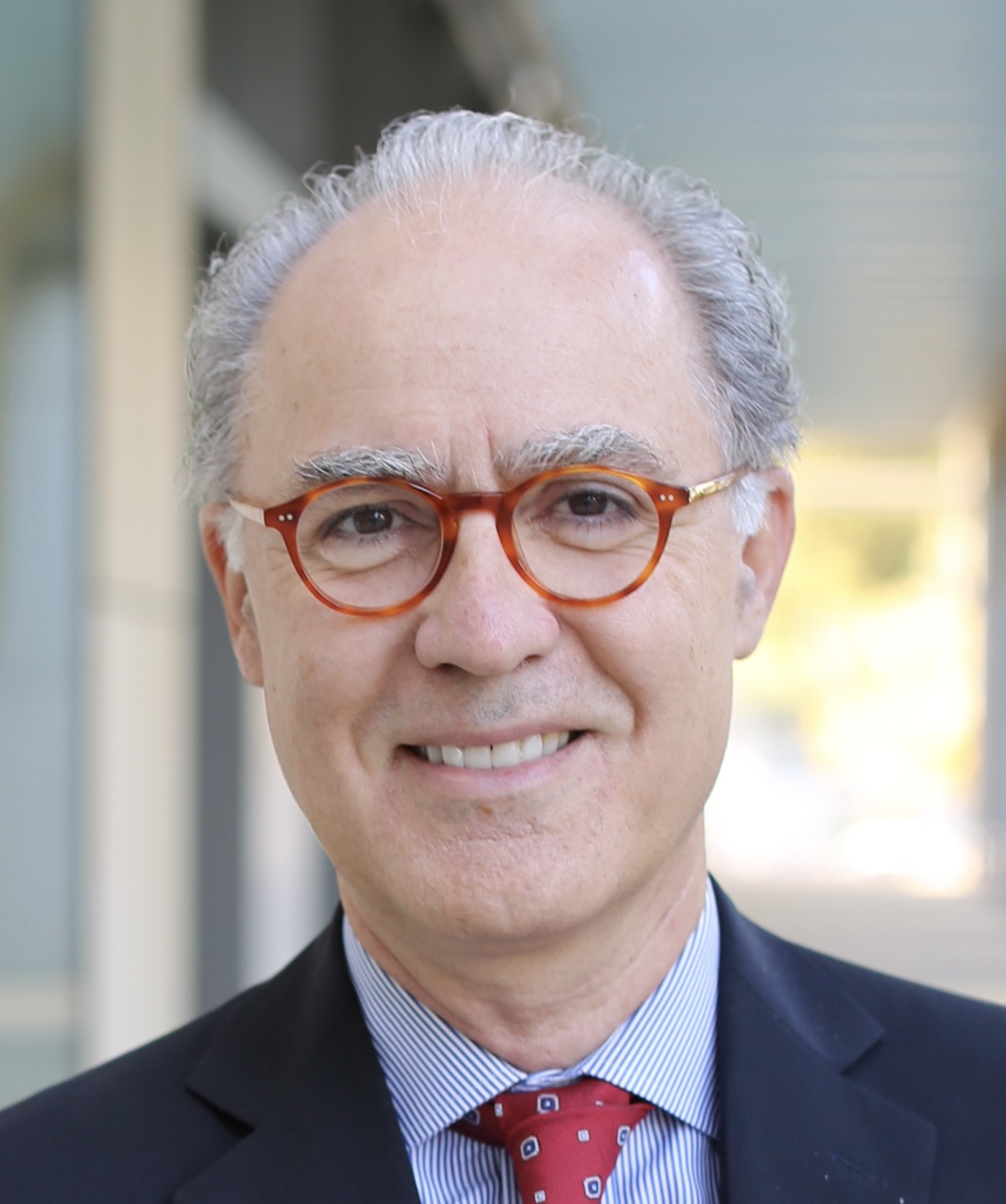}}]{Miroslav Krsti\'c} (Fellow, IEEE) received the Ph.D. degree in electrical engineering from UC Santa Barbara in 1995. He is currently a Distinguished Professor with UC San Diego. Dr. Krsti\'c is Fellow of IFAC, SIAM, AIAA, ASME, AAAS, IET, Serbian Academy of Sciences and Arts, and Academia Europaea. Krsti\'c was the recipient of the IEEE Brocket Control Systems Award,  Bode Lecture Prize, Bellman Control Heritage Award, SIAM Reid Prize, ASME  Oldenburger Medal, John Ragazzini Education Award, four IFAC TC Awards (Nonlinear Control Systems, Distributed Parameter Systems, Adaptive and Learning Systems, Time-Delay Systems), and IFAC Harold Chestnut Control Textbook Award. Krsti\'c is a Fellow-Ambassador of the French CNRS. He serves as the Editor-in-Chief of IEEE Transactions on Automatic Control and has previously served as Editor-in-Chief of Systems and Control Letters. 
\end{IEEEbiography}

\vspace{-.5cm}
\begin{IEEEbiography}[{\includegraphics[width=1in,height=1.25in,clip,keepaspectratio]{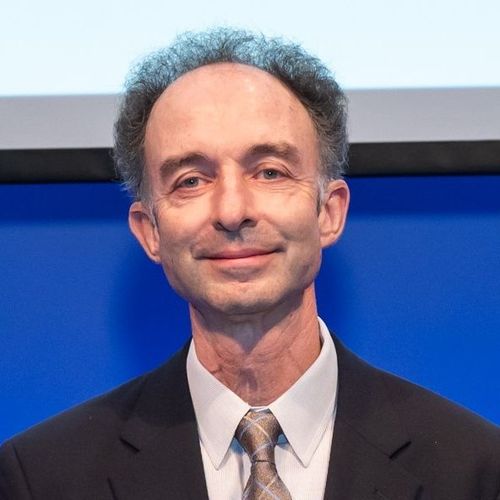}}]{Eugene Lavretsky} (Fellow, IEEE) is a Boeing Principal Senior Technical Fellow within the Boeing Research \& Technology in Huntington Beach, CA.
During his professional career at Boeing, Dr. Lavretsky has developed flight control methods, system identification tools and flight simulation technologies for piloted aircraft and autonomous aerial platforms. His research interests include robust and adaptive control, system identification and flight dynamics. He has written over 100 technical articles, a textbook and taught graduate control courses at the California Long Beach State University, Claremont Graduate University, California Institute of Technology, University of Missouri Science and Technology, and at the University of Southern California. Dr. Lavretsky is the recipient of the AIAA Mechanics and Control of Flight Award (2009), the IEEE Control System Magazine Outstanding Paper Award (2011), the AACC Control Engineering Practice Award (2012), and the IEEE Control Systems Society Award for Technical Excellence in Aerospace Control (2021). Dr. Lavretsky is a Fellow of IEEE (2016), a Fellow of AIAA (2018), and a member of the National Academy of Engineering (2024).
\end{IEEEbiography}

\end{document}